\documentclass[sigconf]{acmart}

\usepackage{multirow}
\usepackage{enumitem}
\setenumerate[1]{itemsep=0pt,partopsep=0pt,parsep=\parskip,topsep=2pt, leftmargin=15pt}
\setitemize[1]{itemsep=0pt,partopsep=0pt,parsep=\parskip,topsep=2pt, leftmargin=10pt}
\usepackage[tight,footnotesize]{subfigure}
\usepackage[linesnumbered,ruled,vlined]{algorithm2e}

\SetCommentSty{mycommfont}
\SetInd{1.25ex}{1.25ex}
\DontPrintSemicolon
\usepackage{stfloats}
\usepackage{color}
\usepackage{url}
\usepackage{hyperref}
\usepackage{mathtools}
\usepackage{booktabs}
\usepackage[braket, qm]{qcircuit}
\usepackage{graphicx}

\usepackage[title]{appendix}

\setenumerate[1]{itemsep=0pt,partopsep=0pt,parsep=\parskip,topsep=5pt}
\setitemize[1]{itemsep=0pt,partopsep=0pt,parsep=\parskip,topsep=5pt}

\newfloat{figtab}{htb}{fgtb}
\makeatletter
  \newcommand\figcaption{\def\@captype{figure}\caption}
  \newcommand\tabcaption{\def\@captype{table}\caption}
\makeatother

\SetKwInput{KwInput}{Input}                
\SetKwInput{KwOutput}{Output}              

\AtBeginDocument{%
  }

\setcopyright{acmlicensed}
\copyrightyear{2025}
\acmYear{2025}
\acmDOI{XXXXXXX}

\acmConference[SIGMOD '25]{SIGMOD 2025}{June 03--05,
  2025}{Woodstock, NY}
\acmISBN{978-1-4503-XXXX-X/18/06}

\begin{document}

\title{Quantum Preference Query}

\author{Hao Liu}
\affiliation{%
  \institution{The Hong Kong University of Science and Technology}
  \country{Hong Kong, China}}
\email{hliubs@cse.ust.hk}

\author{Xiaotian You}
\affiliation{%
  \institution{The Hong Kong University of Science and Technology}
  \country{Hong Kong, China}}
\email{xyouaa@ust.hk}

\author{Raymond Chi-Wing Wong}
\affiliation{%
  \institution{The Hong Kong University of Science and Technology}
  \country{Hong Kong, China}}
\email{raywong@cse.ust.hk}

\begin{abstract}
Given a large dataset of many tuples, it is hard for users to pick out their preferred tuples. Thus, the preference query problem, which is to find the most preferred tuples from a dataset, is widely discussed in the database area. In this problem, a utility function is given by the user to evaluate to what extent the user prefers the tuple. However, considering a dataset consisting of $N$ tuples, the existing algorithms need $O(N)$ time to answer a query, or need $O(N)$ time for a cold start to answer a query. The reason is that in a classical computer, a linear time is needed to evaluate the utilities by the utility function for $N$ tuples. In this paper, we discuss the Quantum Preference Query (QPQ) problem. In this problem, the dataset is given in a quantum memory, and we use a quantum computer to return the answers. Taking the advantage of quantum parallelism, the quantum algorithm can theoretically perform better than their classical competitors. To better cover all the possible study directions, we discuss this problem in different kinds of input and output. In the QPQ problem, the input can be a number $k$ or a threshold $\theta$. Given $k$, the problem is to return $k$ tuples with the highest utilities. Given $\theta$, the problem is to return all the tuples with utilities higher than $\theta$. Also, in QPQ problem, the output can be classical (i.e., a list of tuples) or quantum (i.e., a superposition in quantum bits). Based on amplitude amplification and post-selection, we proposed four quantum algorithms to solve the problems in the above four scenarios. We give an accuracy analysis of the number of memory accesses needed for each quantum algorithm, which shows that the proposed quantum algorithms are at least quadratically faster than their classical competitors. In our experiments, we did simulations to show that to answer a QPQ problem, the quantum algorithms achieve up to 1000$\times$ improvement in number of memory accesses than their classical competitors, which proved that QPQ problem could be a future direction of the study of preference query problems.
\end{abstract}

\maketitle

\section{Introduction}
\label{sec:introduction}
Given a dataset described by several attributes, 
the problem of finding the user's favorite tuples among the dataset is widely involved in many scenarios. For example, purchasing a car, buying a house, and picking a red wine.

Consider the scenario that Alice wants to buy a used car,
where each car is described by some attributes, e.g., price and horsepower. 
Alice may want to purchase a cheap car with a high horsepower.
In the literature, the user preference is usually modeled as a \emph{utility function}.
It quantifies a user's trade-off among attributes and characterizes a vector comprising one weight per attribute. 
Each weight represents the importance of the attribute to a user. 
Based on the utility function, we can obtain a \emph{utility} (i.e., a function score) for each tuple.
The utility indicates to what extent the user prefers the tuple, where a larger utility means that the tuple is more favored by the user. 

Many operators have been proposed to assist users in finding their favorite tuples from a dataset with multiple attributes.
Such operators usually are known as \emph{multi-criteria decision-making tools}.
One representative operator is the preference query \cite{yiu2007top, pivert2012fuzzy, rocha2010efficient, wang2017skyline}. 
The preference query returns tuples with the highest utilities, namely the preferred tuples, w.r.t. the utility function given by a user.
It has many real-world applications such as the recommendation systems \cite{miao2016si2p}, E-commerce \cite{doring2008advanced} and social network \cite{schenkel2008efficient}.

However, given a dataset consisting of $N$ tuples, any preference query algorithm in a classical computer needs $\Omega(N)$ time in total to return the preferred tuples. This observation depends on the fact that a classical algorithm always needs a linear time to calculate the utilities for all the tuples. The efficiency issue leads to the limitation of scalability.

Recently, quantum algorithms have attracted more and more attention. Many quantum algorithms like Shor's algorithm \cite{shor1994algorithms} and Grover's algorithm\cite{grover1996fast} have been proposed and are expected to show quadratic or even exponential speedup compared to classical algorithms. Based on Grover's algorithm, DH algorithm \cite{durr1996quantum} was proposed to find a minimum value from an unsorted table of $N$ values in $O(\sqrt{N})$ time, which is quadratically faster than any classical algorithm. D{\"u}rr et al. \cite{durr2006quantum} proposed a quantum $k$-minima algorithm that returns $k$ smallest values in $O(\sqrt{Nk})$ time. Unfortunately, it only works on a set of integer numbers. Since the typical tuples in a preference query consist of multiple attributes, the existing algorithm cannot be applied on a preference query. Moreover, there is no accurate analysis of the quantum $k$-minima algorithm in \cite{durr2006quantum}, such that we cannot compare the quantum algorithm with a classical algorithm in a real-world scenario.

Motivated by the limitation of the classical preference query and the existing quantum query, we want to combine the strength of both queries.  
Formally, we propose a problem called \emph{\underline{Q}uantum \underline{P}reference \underline{Q}uery (Problem QPQ)}, which finds the user's preferred tuples with multi-attributes with the help of quantum computation. 
Specifically, we return a list of tuples with the highest utilities from a dataset stored in a quantum memory. The utilities are calculated by a utility function given by a user, and we cannot assume any prior knowledge about the utility function. 

There are two distinctive characteristics of our problem QPQ. 
The first characteristic is that it allows tuples to be described by multiple attributes. 
This largely enhances its applicability in daily life scenarios. 
For example, a second-hand car may have three attributes, which are price, horsepower and used mileage. Alice may want to purchase a cheap second-hand car with a high horsepower, and Bob may want to purchase a high-horsepower car with a low used mileage. Since multiple attributes are considered in users' decision-making, the utility function is supposed to allow multiple variables. 

The second characteristic is the utilization of Quantum computation. 
Since the number of quantum states increases exponentially when the number of quantum bits is increasing linearly, the quantum computer is supposed to be more powerful to help a lot of calculations. In recent years, studies like \cite{li2021vsql, jerbi2021parametrized, li2021quantum} in the machine learning field are looking for applications of quantum models. To use variational quantum algorithms, where only the models are running on quantum computers and the optimizers are classical, becomes a popular direction in ML. In the QPQ problem, taking the advantage of quantum parallelism, we are able to calculate the utilities for all the tuples in a single step, which shows that the QPQ algorithm has the potential to answer a preference query more efficiently than a classical algorithm.

Based on real-world observations, we conclude that there are two types of QPQ input. The first one is to obtain the top-$k$ preferred tuples given $k$, which is called \emph{QPQ}$_k$. The other one is to obtain all the preferred tuples with utilities higher than a threshold $\theta$ given $\theta$, which is called \emph{QPQ$_\theta$}. Also, considering the real application in future quantum computers, we discuss two types of QPQ output. The first one is the classical output, which is to return all the preferred tuples in a list. The second one is the quantum output, which is to return quantum bits in a superposition of all the preferred tuples, where superposition is the ability of a quantum system to be in multiple states simultaneously. Therefore, we give a comprehensive discussion of QPQ problem in totally $4$ scenarios, which are almost all the scenarios we will meet in the future.

To the best of our knowledge, we are the first to study problem QPQ. 
There are some closely related studies. $\tau$-LevelIndex \cite{zhang2022t} proposed recently can answer a top-$k$ preference query very efficiently. However, the efficiency deteriorates rapidly when the dimension of the tuple grows. Moreover, it needs extra time to preprocess the data. Quick selection \cite{cormen2022introduction} is also a very popular method to pick out the preferred tuples. However, it needs a linear time to return the answer, which is not efficient enough. 


Our contributions are described as follows.
\begin{itemize}
    \item To the best of our knowledge, we are the first to propose the problem of quantum preference query.
    
    \item We propose algorithms for QPQ problems in 4 different real-world scenarios. We show that their time complexities have asymptotically quadratic improvement over any classical algorithms that require linear time to answer the preference queries. Thus, the time complexities of our algorithms can never be achieved by any classical algorithm.
    
    \item We conducted experiments to demonstrate the superiority of our algorithms. 
    Under typical settings, our algorithms are up to 1000$\times$ faster than classical competitors in number of memory accesses.
\end{itemize}

The rest of the paper is organized as follows. In Section~\ref{sec:preliminary}, we introduce some basic knowledge used in this paper about quantum algorithms. We discuss the related work in Section~\ref{sec:related work}. 
The formal problem definitions and relevant preliminaries are shown in Section~\ref{sec:problem definition}. 
Section~\ref{sec:algorithm} describes algorithms to solve QPQ problems in $4$ different scenarios.
Experiments are shown in Section~\ref{sec:experiments}. 
Section~\ref{sec:conclusion} concludes our paper.

\section{Preliminaries}
\label{sec:preliminary}
To introduce the model of quantum computing, we first recap the following concepts in the classical computing.
\begin{itemize}
	\item \textbf{C1}: When data is stored in memory or on disks, the \emph{bit} is a basic unit.
	\item \textbf{C2}: The bit has two \emph{states}, which are $0$ and $1$.
	\item \textbf{C3}: By a \emph{read} operation, we can obtain the state of a bit.
	\item \textbf{C4}: By a logic \emph{gate}, we can change the state of a bit.
\end{itemize}

The concepts corresponding to C1, C2, C3 and C4 in quantum computing are introduced in the following.

\textbf{C1}: In a quantum computer, we have the \emph{quantum bit}. For simplicity, we call it a \emph{qubit}. Usually, we use $q_{n-1},\cdots, q_0$ (or simply a greek letter says $\psi$) to denote a list of $n$ qubits.

\textbf{C2}: The qubit has \emph{quantum states}. We use the \emph{Dirac notation}~\cite{dirac1939new, bayer2002organization} (i.e., ``$\ket \cdot$'') to represent a quantum state. Inside this Dirac notation, we write an integer 0 or 1 to denote a \emph{basis state}, or write the notation of the qubit to denote a \emph{mixed state}. For example, $\ket 0$ and $\ket 1$ are two basis states of a qubit. For a qubit says $q$, $\ket q$ denotes a mixed state of this qubit (more formally known as \emph{superposition}), which is a state ``between'' $\ket 0$ and $\ket 1$. A mixed state can be represented by a linear combination of the two basis states:
$$\ket q = \alpha \ket 0 + \beta \ket 1$$
where $\alpha$ and $\beta$ are two complex numbers called the \emph{amplitudes}. In quantum computing, we have $|\alpha|^2 + |\beta|^2=1$ for any quantum state, where $|\alpha|^2$ denotes the absolute square of $\alpha$.

\textbf{C3}: We can \emph{measure} a qubit.
The measurement will collapse the superposition of the qubit $q$ so that we can only obtain the result state 0 with probability $|\alpha|^2$ or state 1 with probability $|\beta|^2$. For a concrete example, let us assume that $q$ has the following mixed state:
$$\ket q =  (0.48-0.64i)\ket 0 +  (0.36+0.48i)\ket 1.$$
If we measure $q$, we can obtain 0 with probability $|0.48-0.64i|^2=0.48^2+0.64^2=0.64$ or 1 with probability $|0.36+0.48i|^2=0.36^2+0.48^2=0.36$.
Note that the sum of the two probabilities is 1.

Consider a list $\psi$ of two qubits $q_1, q_0$. We write a list of Dirac notations in order to denote the state of a list of qubits. For instance, it can be observed that the basis states of $\psi$ are $\ket{0}\!\ket{0}$, $\ket{0}\!\ket{1}$, $\ket{1}\!\ket{0}$ and $\ket{1}\!\ket{1}$. If we measure a qubit in a list of multiple qubits, the state of the other qubits could be changed. This phenomenon is called \emph{quantum entanglement} \cite{schrodinger1935discussion}. For example, assume that the two qubits $q_1$ and $q_0$ combining a list of qubits $\psi$ (or simply called a \emph{quantum register}) are represented as
\begin{align*}
\ket{\psi}=\ket{q_1}\!\ket{q_0} &= 0.6\ket 0(0.8 \ket 0 + 0.6 \ket 1)+0.8\ket 1(\frac{1}{\sqrt{2}}\ket 0+\frac{1}{\sqrt{2}}\ket 1)\\
&= 0.48\ket{0}\!\ket{0} + 0.36\ket{0}\!\ket{1} + \frac{0.8}{2}\ket{1}\!\ket{0} + \frac{0.8}{2}\ket{1}\!\ket{1},
\end{align*}
where $\ket{q_1} = 0.6\ket 0 + 0.8\ket 1$ is expressed outside the parentheses, and $\ket{q_0}$ is expressed inside the parentheses and depends on \ket{q_1} due to quantum entanglement. If we measure $q_1$, we will obtain 0 with probability 0.36 or 1 with probability 0.64. If we obtain 0 for $q_1$, the state of $q_0$ will be changed to $\ket{q_0}=0.8 \ket 0 + 0.6 \ket 1$. Otherwise, the state of $q_0$ will be changed to $\ket{q_0}=\frac{1}{\sqrt{2}}\ket 0+\frac{1}{\sqrt{2}}\ket 1$. In a word, the measurement of $q_1$ changes the state of $q_0$. These two entangled qubits have $4$ basis states $\ket{0}\!\ket{0}, \ket{0}\!\ket{1}, \ket{1}\!\ket{0}, \ket{1}\!\ket{1}$ with amplitudes $0.48, 0.36, \frac{0.8}{2}, \frac{0.8}{2}$, respectively. Similarly, we can extend this system to $n$ qubits $q_{n-1}, q_{n-2}, \cdots, q_0$ with $2^n$ basis states and the corresponding $2^n$ amplitudes.

\textbf{C4}: The state of qubits can be transformed by \emph{quantum gates}. For example, \emph{Z-gate} \cite{nielsen2001quantum} is a very useful quantum gate, which turns $\ket 0$ into $\ket{0}$ and turns $\ket 1$ into $-\ket{1}$. Following common approaches, we illustrate the process of quantum transformation with \emph{quantum circuit}. The following shows an example of a quantum circuit which consists of a qubit (denoted by $q$), a classical register (denoted by $c$), a $Z$-gate (denoted by the box containing a ``Z''), and a measurement (denoted by an icon like a real meter).
\begin{center}
	\scalebox{0.8}{
\Qcircuit @C=1.0em @R=1.0em @!R { 
	 	\nghost{{q} :  } & \lstick{{q} :  } & \gate{\mathrm{Z}} & \meter & \qw & \qw\\
	 	\nghost{{c} :  } & \lstick{{c} :  } & \lstick{/_{_{1}}} \cw & \dstick{_{_{\hspace{0.0em}0}}} \cw \ar @{<=} [-1,0] & \cw & \cw\\
}
\hspace{5mm}
}
\end{center}
In this quantum circuit, the \emph{single-line wire} is a timeline representing the process from an earlier moment to a later moment, which denotes the order of the quantum gates coming to the qubit. The \emph{double-line wire} denotes the classical register. The small number below the double-line wire after ``$/$'' is the number of bits in the classical register, so the classical register has $1$ bit in this example. Since there is a small number $0$ under the down arrow from the meter icon, the $0$-th bit in the classical register $c$ stores the result of the measurement.

Here, a $Z$-gate can be represented by the Dirac notation:
$$\ket 0 \rightarrow \ket{0}; \ket 1 \rightarrow -\ket 1, $$
where the two cases of the quantum transformation are separated by ``;'' and ``$\rightarrow$'' denotes the transformation of quantum states.
For example, assume an input $\ket{q}=\frac{1}{\sqrt{2}}\ket{0}+\frac{1}{\sqrt{2}}\ket{1}$. The $Z$-gate turns $\ket 0$ into $\ket{0}$ and turns $\ket 1$ into $-\ket 1$, and thus the quantum bit will become $\ket{q}=\frac{1}{\sqrt{2}}\ket{0}-\frac{1}{\sqrt{2}}\ket{1}$.

The quantum gates can also be applied on multiple qubits. The following quantum circuit shows an example. 
\begin{center}
	\scalebox{0.8}{
\Qcircuit @C=1.0em @R=1.0em @!R { 
	 	\nghost{{q}_{1} :  } & \lstick{{q}_{1} :  } & \qw & \gate{\mathrm{X}} & \meter & \qw & \qw & \qw\\
	 	\nghost{{q}_{0} :  } & \lstick{{q}_{0} :  } & \gate{\mathrm{H}} & \ctrl{-1} & \qw & \meter & \qw & \qw\\
	 	\nghost{{c} :  } & \lstick{{c} :  } & \lstick{/_{_{2}}} \cw & \cw & \dstick{_{_{\hspace{0.0em}0}}} \cw \ar @{<=} [-2,0] & \dstick{_{_{\hspace{0.0em}1}}} \cw \ar @{<=} [-1,0] & \cw & \cw\\
 }
\hspace{5mm}
}
\end{center}
We have two qubits $q_1$ and $q_0$ in this example. We apply a \emph{Hadamard gate} ($H$-gate) \cite{hadamard1893resolution, nielsen2001quantum} on $q_0$ (denoted by the box containing H), and a \emph{controlled-$X$} gate on both $q_0$ and $q_1$ (denoted by the box containing X with one wire as input from the left (i.e., $\ket{q_1}$) and another wire as input from the bottom (i.e., $\ket{q_0}$ after the $H$-gate)). $H$-gate is a very useful quantum gate, which turns $\ket 0$ into $\frac{1}{\sqrt{2}}\ket 0 + \frac{1}{\sqrt{2}}\ket 1$ and turns $\ket 1$ into $\frac{1}{\sqrt{2}}\ket 0 - \frac{1}{\sqrt{2}}\ket 1$, which is represented by the Dirac notation:
$$\ket 0 \rightarrow \frac{1}{\sqrt{2}}\ket 0 + \frac{1}{\sqrt{2}}\ket 1; \ket 1 \rightarrow \frac{1}{\sqrt{2}}\ket 0 - \frac{1}{\sqrt{2}}\ket 1.$$ The $X$-gate (also known as a NOT-gate) is to swap the amplitudes of $\ket 0$ and $\ket 1$. We can represent it by the Dirac notation: $$\ket 0 \rightarrow \ket 1; \ket 1 \rightarrow \ket 0.$$ We use ``$X \ket{q_1}$'' to denote applying an $X$-gate on $q_1$.
The controlled-$X$ gate has two parts: $q_0$ is the \emph{control} qubit since it is in the wire from the bottom, and $q_1$ is the \emph{target} qubit since it is in the wire from the left. The controlled-$X$ gate applies an $X$-gate on the target qubit if the control qubit is $\ket 1$; otherwise, it remains the target qubit unchanged. We can represent it by the Dirac notation: $$\ket{q_1}\ket{0}\rightarrow \ket{q_1}\ket{0}; \ket{q_1}\ket{1}\rightarrow X\ket{q_1}\ket{1}.$$

We describe such a quantum transformation consisting of a series of quantum gates as a \emph{quantum oracle}. When we do not need to focus on the details in the quantum circuit but only focus on the quantum transformation, we can use a quantum oracle to express such a quantum transformation. The concept of the quantum oracle has been widely used in many studies such as \cite{grover1996fast, shor1994algorithms, wiebe2015quantum, zhang2018quantum}. Since different time complexities can be obtained on different gate sets (which represent the instruction sets in classical computers) and the general quantum computer is still at a very early stage, the query complexity is used to study quantum algorithms, which is to measure the number of queries to the quantum oracles. In the quantum algorithm area, many studies such as \cite{li2021sublinear, kapralov2020fast, montanaro2017quantum, naya2020optimal, hosoyamada2018quantum, li2019sublinear, kieferova2021quantum} assume that quantum oracles costs $O(1)$ time, and analyze the time complexity based on this assumption. We also follow this common assumption in this paper.


\section{Related Work}
\label{sec:related work}

Given a list of $N$ tuples, without any preprocessing, the classical algorithm \emph{Quick Selection}~\cite{cormen2022introduction} is used to select the $k$ tuples with the highest utilities (i.e., the top-$k$ tuples). If preprocessing is allowed, a data structure can be built to speedup the queries. \emph{$\tau$-Level Index}~\cite{zhang2022t} was proposed to answer a top-$k$ query efficiently. The preprocessing algorithm scans the hyperplane of all the possible utility functions, divides them into several partitions recursively, and builds a tree data structure, where each tree node of level $i$ contains the tuple with the $i$-th highest utility in the hyperplane. By the experiments in \cite{zhang2022t}, $\tau$-Level Index is the state-of-the-art to answer a top-$k$ query. However, the preprocessing time is long and the utility function is limited to linear function and thus cannot be a general function. Obviously, the whole execution time to answer a preference query is at least $\Omega(N)$ in a classical computer, because a linear time is needed to calculate the utilities for all the tuples. Therefore, the quantum search algorithm can be considered to improve the efficiency.

In quantum computing, Grover's algorithm \cite{grover1996fast} is described as a database search algorithm. It solves the problem of searching a record in an unstructured list, where all the $N$ records are arranged in random order. On average, the classical algorithm needs to perform $N/2$ queries to a function $f$ to tell us if the record is the answer. More formally, for each index $x$, $f(x)=1$ means the record is the answer and $f(x)=0$ means the record is not the answer. If we have a quantum circuit to calculate this function, then we can build a Grover oracle $G:\ket{x} \rightarrow (-1)^{f(x)}\ket{x}$. Taking the advantage of quantum parallelism, Grover's algorithm can find the index of the answer with $O(\sqrt{N})$ queries to the oracle. The main idea is to first ``flip'' the amplitude of the answer state and then reduce the amplitudes of the other states. One such iteration will enlarge the amplitude of the answer state and $O(\sqrt{N})$ iterations need to be performed until the probability that the qubits are measured to be the right answer is close to 1. Grover mentioned in \cite{grover2005partial} that the database is supplied in the form of a quantum oracle. In fact, the Grover oracle also contains the information of the query condition, but not only a database. This oracle can recognize the solution and Grover's algorithm is from ``recognizing the solution'' to ``knowing the solution'' \cite{nielsen2001quantum}, so Grover's algorithm has limitations in database searching. For example, the generation of the Grover oracle can be even slower than classical search \cite{seidel2021automatic}. However, many algorithms invoke Grover's algorithm as a subroutine due to its quadratic speedup compared to classical algorithms.

Dürr and Høyer's algorithm (DH) \cite{durr1996quantum} is one of the well-known applications of Grover's algorithm and has also become a subroutine of a lot of algorithms such as \cite{wiebe2015quantum}. DH algorithm is to find the index of the minimum record in an unsorted table. The main idea is to randomly choose a record and use Grover's algorithm to randomly choose a smaller record iteratively. The authors proved that with the total running time is less than $22.5\sqrt{N}+1.4\lg^2N$ and the algorithm can find the index of the minimum record with probability at least $\frac{1}{2}$. Compared with the quantum algorithm, the classical algorithm needs at least $\frac{N}{2}$ time. Another example was shown in \cite{grover1996fastm}. To estimate the median $\mu$ of $N$ items in an unordered list with a precision $\varepsilon$ such that both the number of records smaller than and greater than $\mu$ are at most $\frac{N}{2}(1+|\varepsilon|)$, any classical algorithm needs to sample at least $\Omega(\frac{1}{\varepsilon^2})$ times. An $O(\frac{1}{|\varepsilon|})$ step quantum algorithm was proposed in this paper, which takes the same phase-shifting method as Grover's algorithm and can give an estimate of $\varepsilon$ given $\mu$. Combined with binary search, this algorithm can be used to find the median. 

In \cite{grover2005partial}, Grover et al. proposed the quantum partial search algorithm. This problem has a same condition as  \cite{grover1996fast} that an unstructured list is given with a function $f$ such that $f(x)=1$ for a unique index $x$. Different from  \cite{grover1996fast}, quantum partial search only needs to find a part of the answer. Specifically, if the index $x$ is an $n$-bit address, then we only need to find the first $k$ bits of $x$. We can regard the first $k$ bits as a block, so we need to find the target block instead of the target record in the original problem. The authors concluded that the partial search is easier than the exact search. The best randomized partial search algorithm is expected to find the answer with $\frac{N}{2}(1-\frac{1}{K^2})$ queries, where $K=2^k$. This algorithm saved $\frac{N}{2K^2}$ queries compared to the original problem. In \cite{grover2005partial}, the authors proposed a better quantum algorithm that saves $\theta(\frac{1}{\sqrt{K}})$ of all queries, which is also asymptotically optimal. The main idea is to divide the original search procedure into global search and local search. We first do some iterations of global search, then do some iterations of local search in all $K$ blocks. Note that by quantum parallelism, the local search is in parallel, so this method is slightly faster than the exact search. Zhang et al. \cite{zhang2018quantum} discussed a harder version of the quantum partial search. In the new problem, we have multiple target records and also multiple target blocks. The target records are unevenly distributed in the list, which means that the target blocks have different numbers of target records. They solved this problem with the same main idea and the algorithm runs the fastest when target records are evenly distributed.

Wiebe et al. proposed a quantum nearest-neighbor algorithm in \cite{wiebe2015quantum} based on DH algorithm, which shows that the quantum algorithm may provide applications to machine learning. The task is to find the closest vector to $u$ in the training data. The training data contains $M$ vectors $v_1, v_2, \cdots, v_M$. Then, the algorithm needs two oracles $\mathcal{O}: \ket j \ket i \ket 0 \rightarrow \ket j \ket i \ket{v_{ji}}$ and $\mathcal{F}: \ket j \ket l \rightarrow \ket j \ket{f(j, l)}$, where $v_{ji}$ is the $i$-th element of the $v_j$ and $f(j, l)$ is the location of the $l$-th non-zero element in $v_j$. Using a subroutine consisting of these two quantum oracles, we can obtain $\frac{1}{\sqrt{M}}\sum_j \ket j (\sqrt{1-|v_j-u|}$ $\ket 0 + \sqrt{|v_j-u|}\ket 1)$. Since $|v_j-u|$ is the distance between $u$ and $v_j$, we have all the distances encoded in the amplitudes. However, since the amplitudes can only deliver the probabilities, we cannot simply read the amplitudes. In the next step, they use amplitude estimation \cite{brassard2002quantum} to estimate the probabilities and store them as states, so we obtain $\frac{1}{\sqrt{M}}\sum_j \ket j \ket{|v_j-u|}$. The final step is to use DH algorithm to find the minimum, and then we know the index of the closest vector.

There are also many studies on quantum machine learning. Li et al. \cite{li2021quantum} proposed a neural network method for conversational emotion recognition. This work does not use the quantum circuit and quantum bits but leverage the quantum algorithm. In particular, they use a complex number vector to encode the three kinds of data and regarded the amplitudes as the probabilities of all the emotions. They proposed a quantum-like operation to update the vector iteratively. Li et al. \cite{li2021vsql} proposed a quantum-classical framework for quantum learning. In general, a machine learning framework has data, a model, a cost function, and an optimizer. Only part of their model is running on a quantum circuit. They used a quantum circuit to extract classical features and then use a fully-connected neural network to do the classification. Jerbi et al. \cite{jerbi2021parametrized} gave a quantum framework in reinforcement learning.  They also used a quantum model and a classical optimizer. The quantum model is based on a parameter $\theta=(\phi, \lambda)$ where $\phi$ is the rotation angle and $\lambda$ is the scaling parameter, and then they used sample interactions and policy gradients to update this policy parameter.

\section{Problem Definition}
\label{sec:problem definition}
In this section, we give a formal definition of quantum preference query problems. We first give a introduction of the quantum memory 
in Section \ref{QRAM}, which is used to store the dataset in quantum computers. Then, we define the quantum preference query problems in Section \ref{qtkdef}.

\subsection{Quantum Random Access Memory}\label{QRAM}
Following \cite{kerenidis2017quantum, saeedi2019quantum}, in this paper, we assume a \emph{classical-write quantum-read QRAM}. It stores a classical record in $O(1)$ time, and accepts a superposition of addresses and returns a superposition of the corresponding records in $O(1)$ time.  It takes $O(1)$ time to read multiple records, which is the main difference from a classical memory. Specifically, we have the following Definition \ref{def:QRAM}.

\begin{definition}[Quantum Random Access Memory (QRAM)]\label{def:QRAM}
A QRAM $\mathcal{Q}$ is an ideal model which performs store and load operations in $O(1)$ time. 
\begin{itemize}
	\item A store operation: $\mathcal{Q}[addr]=value$, where $addr$ and $value$ are two bit-strings;
	\item A load operation: $$\mathcal{Q}\frac{1}{\sqrt{M}}\sum_{i=0}^{M-1}\ket{addr_i}\ket{x_i}=\frac{1}{\sqrt{M}}\sum_{i=0}^{M-1}\ket{addr_i}\ket{x_i\oplus value_i},$$ where $M$ is the number of required values, $addr_i$ and $value_i$ are the bit-strings denoting the $i$-th address and the value stored at address $addr_i$, respectively, $x_i$ is the initial value of the $i$-th destination register, and $\oplus$ denotes the XOR operation.
\end{itemize}
\end{definition}

Assume there is a QRAM $\mathcal{Q}$, then $\mathcal{Q}$ supports two kinds of operations. The first is to store classical data. An operation $\mathcal{Q}[addr]=value$ stores a classical value at the specified address in $O(1)$ time, where $addr$ and $value$ are two bit-strings denoting the address and value, respectively.
For convenience, we also write $\mathcal{Q}[addr]$ on the right-hand side of an assignment expression to denote the value stored at the address $addr$.
The second operation is to load the quantum data. In this operation, $\mathcal{Q}$ is a quantum mapping from the addresses to the values: $\mathcal{Q}\ket{addr}\ket{x}=\ket{addr}\ket{x\oplus value}$, where $\oplus$ denotes the XOR operation and $\ket x$ is the initial state of the returned qubits. Furthermore, the loading operation
$$\mathcal{Q}\frac{1}{\sqrt{M}}\sum_{i=0}^{M-1}\ket{addr_i}\ket{0}=\frac{1}{\sqrt{M}}\sum_{i=0}^{M-1}\ket{addr_i}\ket{value_i}$$
costs $O(1)$ time (note that $0 \oplus value_i = value_i$ in the above equation). In the database area, the load and store operation cost $O(1)$ time is a common basic assumption. In the quantum algorithm area, this quantum mapping can be regarded as a quantum oracle, and many studies such as \cite{li2021sublinear, kapralov2020fast, montanaro2017quantum, naya2020optimal, hosoyamada2018quantum, li2019sublinear, kieferova2021quantum} also assume the quantum oracle costs $O(1)$ time.

Many quantum algorithms \cite{kerenidis2019q, rebentrost2014quantum, wiebe2012quantum, kapoor2016quantum, wiebe2015quantum} use the concept of a QRAM. The reason is that many quantum algorithms were proposed to have a sub-linear time complexity compared to their linear classical competitors. If loading the data needs a linear time, this advantage may be lost from the theoretical perspective. For example, consider the Grover's algorithm \cite{grover1996fast} introduced in Section \ref{sec:related work}. If we want to use Grover's algorithm as a common database search algorithm, we need a quantum random access memory (QRAM) to store all the records. Assume the query is to find the position of the word ``unicorn'' in \emph{The Witcher} and the word only appears once in the book. If all the $N$ words are stored in the QRAM, then the QRAM can do the following quantum operation efficiently: $$\frac{1}{\sqrt{N}}\sum_{i=0}^{N-1}\ket{i}\ket{0}\rightarrow \frac{1}{\sqrt{N}}\sum_{i=0}^{N-1}\ket{i}\ket{w_i},$$where $N$ is the total number of words in the book and $w_i$ is the $i$-th word in the book. Then, we can build a quantum oracle that flips the amplitude of $\ket{i}$ if $w_i$ is ``unicorn''. With this oracle, we can find the position of the word in $O(\sqrt{N})$ time. However, without a QRAM, we need to use a linear time to load all the words in the book, then, theoretically, Grover's algorithm will lose the advantage of quantum parallelism.
 
Note that even with the assumption of QRAMs, the quantum preference query is still non-trivial. The reason is that the QRAM only maps the addresses to values, but to fetch all the addresses of the desired results is a non-trivial task.
 
\subsection{Quantum Preference Query}\label{qtkdef}

In the problem, a dataset $D=\{p_0, p_1,\cdots, p_{N-1}\}$ of $N$ tuples is given. For each $i\in[0, N-1]$, the tuple $p_i$ is represented as a $d$-dimensional points, i.e., $p_i=(p_i[0], p_i[1],\cdots, p_i[d-1])$, where $d$ is the number of dimensions and each dimension corresponds to an attribute of the tuple. In the quantum problem, we assume that the dataset is given in a QRAM $\mathcal{Q}$ where $\mathcal{Q}[i]=p_i$ following \cite{grover1996fast, durr1996quantum, grover2005partial}. Let $n$ be $\lceil \log_2 N \rceil$. It is easy to observe that each index from 0 to $N-1$ can be represented by $n$ qubits.

Then, a \emph{utility function} denoted by $f(\cdot)$ is input by the user. It takes a tuple (says $p$) as input and returns a real number (i.e., $f(p)$) to represent the utility of this tuple. We do not need the common assumption that the utility function is linear (i.e., $f(p)=\sum w[i]p[i]$ where $w[i]$ is the weight of $p[i]$ given by the user). Thus, for instance, $f(p)=\sqrt{\sum p[i]^2}$ can also be one of the options chosen by the user. Note that by \cite{toffoli1980reversible}, we can always build a quantum circuit with only Toffoli gates to calculate an algebraic expression, so a quantum oracle $\mathcal{F}$ can be built based on the utility function given by the user, such that $$\mathcal{F}\ket{p}\ket{x}=\ket{p}\ket{x\oplus f(p)},$$ which means we can calculate the utilities efficiently in a quantum computer. For simplicity, we assume the attributes and utilities can be represented as integers. Note that it is trivial to extend integers to real numbers in quantum computers, which is the same as classical computers. Specifically, we assume each attribute can be represented by an $n_a$-bit integer, and each utility can be represented by an $n_u$-bit integer.

Given a dataset $D$ (stored in a QRAM $\mathcal{Q}$ in a quantum computer) and a utility function $f$ (represented as a quantum oracle $\mathcal{F}$ in a quantum computer), our task is to find tuples with the highest utilities. To be more comprehensive, we consider both different types of input and different types of output. In literature, two types of input are considered:
\begin{itemize}
	\item A threshold $\theta$. Following \cite{gallego2019threshold, regenwetter1998random, liang2011utility}, the query returns all tuples with utilities higher than $\theta$.
	\item A size $k$. Following \cite{lee2009personalized, peng2015k, lian2009top, Soliman2009}, the query returns $k$ tuples with the highest utilities.
\end{itemize} 
Also, two types of output are considered:
\begin{itemize}
	\item Classical output. A list of $k$ tuples $\{p_{l_0}, p_{l_1},\cdots, p_{l_{k-1}}\}$ is returned, where $l_i$ is the index of the $i$-th tuple in the list.
	\item Quantum output. Following \cite{harrow2009quantum, rebentrost2014quantum, coppersmith2002approximate}, a quantum query returns quantum bits $\ket{\psi}$ in a superposition of the desired $k$ tuples such that $$\ket{\psi}=\frac{1}{\sqrt{k}}\sum_{i=0}^{k-1} \ket{p_{l_i}}.$$
\end{itemize} 
Therefore, four different problems are proposed, which covers all the cases both in the classical world and in the quantum world.

\begin{definition}[Classical-Output Threshold-Based Quantum Preference Query (CQPQ$_\theta$)] \label{def:ctqk}
Given a dataset $D$, a utility function $f$ and a threshold $\theta$, we would like to find out all the tuples with utility higher
than or equal to $\theta$ and return them in classical bits.
\end{definition}

\begin{definition}[Quantum-Output Threshold-Based Quantum Preference Query (QQPQ$_\theta$)] \label{def:tqk}
Given a dataset $D$, a utility function $f$ and a threshold $\theta$, we would like to find out all the tuples with utility higher
than or equal to $\theta$ and return them in quantum bits.
\end{definition}

\begin{definition}[Classical-Output Top-$k$ Quantum Preference Query (CQPQ$_k$)] \label{def:cqk}
Given a dataset $D$, a utility function $f$ and a positive integer $k$, we would like find out the $k$ tuples with the highest 
utility and return them in classical bits.
\end{definition}

\begin{definition}[Quantum-Output Top-$k$ Quantum Preference Query (QQPQ$_k$)]\label{def:qk}
Given a dataset $D$, a utility function $f$ and a positive integer $k$, we would like to find out the $k$ tuples with the highest 
utility and return them in quantum bits.
\end{definition}

Combining $\theta$ input and quantum output, we give Definition \ref{def:tqk} of Quantum-output threshold-based quantum preference query. Combining $\theta$ input and classical output, we give Definition \ref{def:ctqk} of classical-output threshold-based quantum preference query. Combining $k$ input and classical output, we give Definition \ref{def:cqk} of classical-output top-$k$ quantum preference query. Combining $k$ input and quantum output, we give Definition \ref{def:qk} of quantum-output top-$k$ quantum preference query.

\section{Algorithm}
\label{sec:algorithm}
In this section, we propose how to solve the four problems introduced in Section \ref{sec:problem definition}. In Section \ref{sec:QQPQT}, we introduce the step to obtain the answer to a QQPQ$_\theta$ problem. Based on QQPQ$_\theta$, we present how to answer a CQPQ$_\theta$ in Section \ref{sec:CQPQT}. Then, we introduce the algorithm for CQPQ$_k$ in Section \ref{sec:CQPQK}. Finally, the algorithm for QQPQ$_k$ is proposed in Section \ref{sec:QQPQK}.

\subsection{Quantum-Output Threshold-Based Quantum Preference Query}\label{sec:QQPQT}

There are two main steps to answer a QQPQ$_\theta$ problem. For simplicity, we slightly abuse notation $k$ by assuming that there are $k$ tuples with utilities higher than $\theta$. The first step is to use amplitude amplification \cite{brassard2002quantum} to enlarge the probability that we can measure the $k$ tuples. The second step is to use post-selection to obtain the $k$ tuples. In Section \ref{sec:AA}, we introduce how to use amplitude amplification with a QRAM in preference query. In Section \ref{sec:post}, we introduce how to use post-selection to obtain the answer. In Section \ref{sec:ana}, we analyze how to merge the two sub-algorithms to obtain the quantum answer to a preference query.

\subsubsection{Amplitude Amplification}\label{sec:AA}
\leavevmode\\
Consider a quantum register $\psi$ consisting of $n+n_u$ qubits with quantum state $\ket{\psi}=\frac{1}{\sqrt{N}}\sum_{i=0}^{N-1}\ket{i}\ket{f(p_i)}$. We denote a quantum state formed by the $k$ tuples whose utilities are higher than $\theta$ as $\ket{\theta^+}=\frac{1}{\sqrt{k}}\sum_{f(p_i)\geq \theta}\ket{i}\ket{f(p_i)}$ and we also denote a quantum state formed by the other $N-k$ tuples as $\ket{\theta^-}=\frac{1}{\sqrt{N-k}}\sum_{f(p_i)<\theta}\ket{i}\ket{f(p_i)}$. Then, we can obtain $\ket{\psi}=\frac{\sqrt{k}}{\sqrt{N}}\ket{\theta^+} + \frac{\sqrt{N-k}}{\sqrt{N}}\ket{\theta^-}$. The major idea of amplitude amplification is to amplify the amplitude of $\ket{\theta^+}$ and reduce the amplitude of $\ket{\theta^-}$ so that we can measure the top-$k$ tuples with a higher probability.

Starting from a set of qubits with a state $\ket{0}$, the algorithm mainly contains three steps. Step 1 is to equalize the amplitudes of all the states from $\ket{0}$ to $\ket{N-1}$. Step 2 is to flip the amplitudes of the top-$k$ tuples from $\frac{1}{\sqrt{N}}$ to $-\frac{1}{\sqrt{N}}$. Step 3 is to selectively rotate different states to enlarge the amplitude of the correct answer. 

In Step 1, we apply Hadamard gates on $n$ qubits with an initial state $\ket{0}$ to denote the index of the $N$ tuples, so we obtain $\frac{1}{\sqrt{N}}\sum_{i=0}^{N-1}\ket{i}$. By the first step, we evenly distribute the probability to all the tuples so that if we measure the qubits we will obtain a number from $0$ to $N-1$ with probability $\frac{1}{N}$.

In Step 2, we need to use the QRAM $\mathcal{Q}$ to read out all the tuples. After appending $d$ quantum registers, we obtain $$\mathcal{Q}\frac{1}{\sqrt{N}}\sum_{i=0}^{N-1}\ket{i}\ket{0}\cdots\ket{0}=\frac{1}{\sqrt{N}}\sum_{i=0}^{N-1}\ket{i}\ket{p_i[1]}\ket{p_i[2]}\cdots\ket{p_i[d]}.$$ Then, based on the utility function $f$, we can construct a quantum oracle $\mathcal{F}$ to calculate the utilities in quantum bits, such that $$\mathcal{F}\ket{p_i[1]}\ket{p_i[2]}\cdots\ket{p_i[d]}\ket{0}=\ket{p_i[1]}\ket{p_i[2]}\cdots\ket{p_i[d]}\ket{f(p_i)}.$$ Appending another quantum register, we obtain 
\begin{align*}
 &\mathcal{F}\frac{1}{\sqrt{N}}\sum_{i=0}^{N-1}\ket{i}\ket{p_i[1]}\ket{p_i[2]}\cdots\ket{p_i[d]}\ket{0}\\
 =&\frac{1}{\sqrt{N}}\sum_{i=0}^{N-1}\ket{i}\ket{p_i[1]}\ket{p_i[2]}\cdots\ket{p_i[d]}\ket{f(p_i)}.
\end{align*}
Then, we can perform an ``non-computation'' trick to reuse the middle $dn_a$ qubits by applying a QRAM read again and disentangling the middle $dn_a$ qubits, so we obtain $\frac{1}{\sqrt{N}}\sum_{i=0}^{N-1}\ket{i}\ket{f(p_i)}$. Then, based on the given threshold $\theta$, we design a quantum oracle $\mathcal{G}_\theta$, such that $$\mathcal{G}_{\theta}\ket{x} = 
\left\{
	\begin{array}{ll}
		-\ket{x} &\mbox{$x\geq \theta$} ;\\
		\ket{x} &\mbox{$x < \theta$} ,\\
	\end{array}
\right.$$ which means that $\mathcal{G}_{\theta}$ will multiply the amplitude of the state larger than $\theta$ by $-1$ and keep other amplitudes unchanged. We use $\mathcal{G}_{\theta}$ to flip the amplitude of the indexes of desired tuples so that we obtain 
\begin{align*}
&\frac{1}{\sqrt{N}}\sum_{i=0}^{N-1}\ket{i}\mathcal{G}_\theta\ket{f(p_i)}\\ =& -\frac{1}{\sqrt{N}}\sum_{f(p_i)\geq \theta}\ket{i}\ket{f(p_i)} +\frac{1}{\sqrt{N}}\sum_{f(p_i)< \theta}\ket{i}\ket{f(p_i)}\\
 =& -\frac{\sqrt{k}}{\sqrt{N}}\ket{\theta^+} +\frac{\sqrt{N-k}}{\sqrt{N}}\ket{\theta^-}.
\end{align*}

In Step 3, we need to further apply the diffusion transform \cite{grover1996fast} $H\mathcal{R}H$ to the qubits, where the quantum oracle $\mathcal{R}$ is to flip the amplitudes of all the states except $\ket{0}$. That is, $$\mathcal{R}\ket{x} = 
\left\{
	\begin{array}{ll}
		-\ket{x} &\mbox{$x \neq 0$} ;\\
		\ket{x} &\mbox{$x=0$}.\\
	\end{array}
\right.$$
By the diffusion transform, the amplitude of $\ket{\theta^+}$ is amplified and the amplitude of $\ket{\theta^-}$ is reduced, which means we will obtain the desired tuple with a higher probability if we measure the qubits. 

It is worth mentioning that Step 2 and Step 3 can be performed iteratively to further amplify the amplitude of $\ket{\theta^+}$. As shown in \cite{boyer1998tight}, after $s$ iterations of Step 2 and Step 3, we obtain $$\sin((2s+1)t)\ket{\theta^+} +\cos((2s+1)t)\ket{\theta^-},$$ where $t=\arcsin \frac{\sqrt{k}}{\sqrt{N}}$. Each iteration performs a QRAM read operation, which corresponds to one IO, so the time complexity depends on the number of iterations.

\begin{figure}[htbp]
  \centering
  \subfigure{
  \includegraphics[width=0.75\columnwidth]{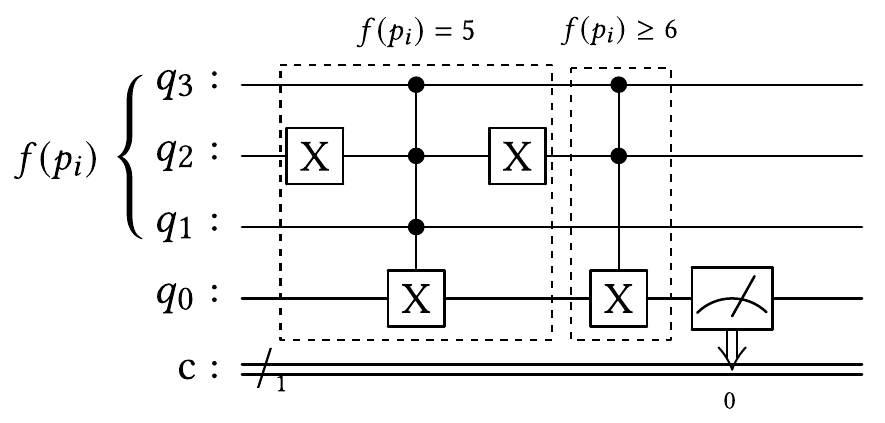}}
  \caption{Post-selection on $\mathcal{O}_5$}
  \label{fig:postexp}
\end{figure}

\subsubsection{Post-Selection}\label{sec:post}
\leavevmode\\

In this section, we propose how to obtain the answer $\ket{\theta^+}$ from the resulting quantum state of amplitude amplification $\sin((2s+1)t)\ket{\theta^+} +\cos((2s+1)t)\ket{\theta^-}$. For simplicity, we use $\alpha\ket{\theta^+} +\beta\ket{\theta^-}$ to denote this quantum state, where $\alpha$ is defined to be $\sin((2s+1)t)$ and $\beta$ is defined to be $\cos((2s+1)t)$.

The first step is to append an ``auxiliary'' qubit $\ket{0}$. Thus, we obtain $\alpha\ket{\theta^+}\ket{0} +\beta\ket{\theta^-}\ket{0}$. We expand the expression to obtain $$\frac{\alpha}{\sqrt{k}}\sum_{f(p_i)\geq \theta}\ket{i}\ket{f(p_i)}\ket{0} +\frac{\beta}{\sqrt{N-k}}\sum_{f(p_i)<\theta}\ket{i}\ket{f(p_i)}\ket{0}.$$ 

Then, we design another quantum oracle $\mathcal{O}_\theta$, which is applied on the last $n_u+1$ qubits and flips the last auxiliary qubit if $f(p_i)\geq \theta$. Specifically, this quantum oracle works as follows. $$\mathcal{O}_\theta\ket{f(p_i)}\ket{x} = 
\left\{
	\begin{array}{ll}
		\ket{x\oplus 1} &f(p_i)\geq \theta;\\
		\ket{x} &\mbox{otherwise} .\\
	\end{array}
\right.$$ It is easy to observe that $\mathcal{O}_\theta$ is similar to $\mathcal{G}_\theta$ in Section \ref{sec:AA}, since they both change the state of the desired tuples. The difference is that $\mathcal{O}_\theta$ uses controlled $X$-gates to flip the last auxiliary qubit and $\mathcal{G}_\theta$ uses controlled $Z$-gates to flip the amplitudes. By applying the quantum oracle $\mathcal{O}_\theta$ on the last $n_u+1$ qubits, we obtain $$\frac{\alpha}{\sqrt{k}}\sum_{f(p_i)\geq \theta}\ket{i}\ket{f(p_i)}\ket{1} +\frac{\beta}{\sqrt{N-k}}\sum_{f(p_i)<\theta}\ket{i}\ket{f(p_i)}\ket{0},$$ which can be represented by $\alpha\ket{\theta^+}\ket{1} +\beta\ket{\theta^-}\ket{0}$. 

In the last step, we measure the last auxiliary qubit. If we obtain $1$, then $\ket{\theta^-}$ will be collapsed in the superposition, such that we obtain $\ket{\theta^+}$, which is the desired tuples in quantum bits. Otherwise, the post-selection fails. The success rate depends on the amplitudes of $\ket{\theta^+}$ and $\ket{\theta^-}$, which is equal to $\frac{\alpha^2}{\alpha^2+\beta^2}$. 

Figure \ref{fig:postexp} shows an example of post-selection. We use $q_3q_2q_1$ to denote the binary representation of $f(p_i)$ and we use $q_0$ to denote the auxiliary qubit. Assume we have $\ket{q_3q_2q_1}=\frac{1}{2}(\ket{0}+\ket{1}+\ket{5}+\ket{7})$ and we are implementing $\mathcal{O}_5$. We add $q_0$ to the last, and then apply the quantum circuit in Figure \ref{fig:postexp}. In the first dashed-line box, we consider the case $f(p_i)=5$ such that we obtain $\ket{q_3q_2q_1}\ket{q_0}=\frac{1}{2}\ket{5}\ket{1}+\frac{1}{2}(\ket{0} + \ket{1}+\ket{7})\ket{0}$. In the second dashed-line box, we consider the case $f(p_i)\geq 6$ such that we obtain $\ket{q_3q_2q_1}\ket{q_0}=\frac{1}{2}(\ket{5}+\ket{7})\ket{1}+\frac{1}{2}(\ket{0} + \ket[1])\ket{0}$. Then, if we measure $q_0$ to be $1$ with probability $\frac{1}{2}$, then we can obtain $\ket{q_3q_2q_1}=\frac{1}{\sqrt{2}}(\ket{5}+\ket{7})$.

\begin{figure*}[htbp]
\centering
\subfigure{
  \includegraphics[width=0.75\textwidth]{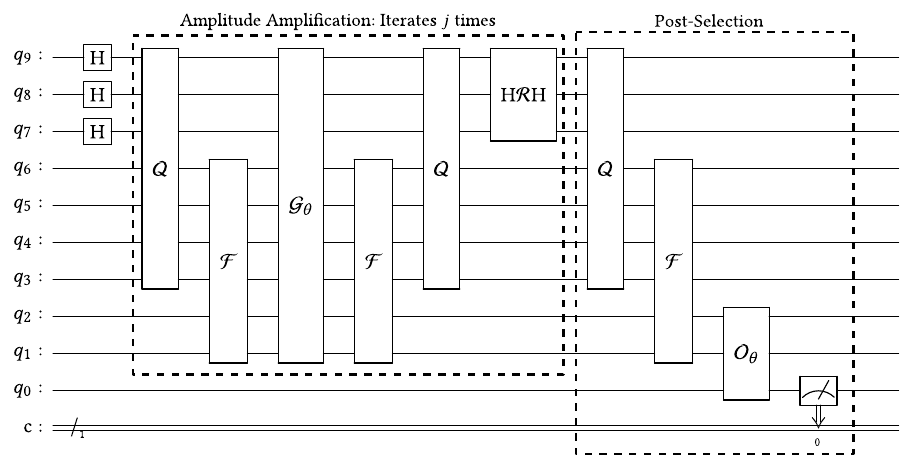}}
\caption{An illustration of the quantum circuit for QQPQ$_\theta$}\label{fig:QQPQ}
\end{figure*}

\begin{algorithm}[htbp]
\SetKwFunction{QTK}{QQPQ$_\theta$}
\SetKwFunction{Break}{break}
\caption{\protect\QTK{f, $\theta$}} \label{alg-1}
\KwIn{A utility function $f$ and a utility threshold $\theta$.}
\KwOut{$\frac{1}{\sqrt{k}}\sum_{f(p_i)\geq \theta}\ket{i}\ket{f(p_i)}$ where $k$ is the number of tuples with utilities higher than $\theta$.}
	$m$ $\leftarrow$ $1$;\\
	\While{$m \leq \sqrt{N}$}{
		$j$ $\leftarrow$ a random number in $\{1,\cdots, m\}$;\\
		$\ket{\psi}$ $\leftarrow$ $\ket{0}$;\\
		$\ket{\psi}$ $\xmapsto{H-gates}$ $\frac{1}{\sqrt{N}}\sum_{i=0}^{N-1} \ket{i}$;\\
		\tcp*[l]{Amplitude amplification}
		\For{$s$ $\leftarrow$ $1$ \KwTo $j$}{
			\tcp*[l]{$t=\arcsin{\frac{\sqrt{k}}{\sqrt{N}}}$ and $k$ is unknown yet}
			$\ket{\psi}$ $=$ $\sin((2s+1)t)\frac{1}{\sqrt{k}}\sum_{f(p_i)\geq \theta}\ket{i} +\cos((2s+1)t)\frac{1}{\sqrt{N-k}}\sum_{f(p_i)< \theta}\ket{i}$;\\
			$\ket{\psi}$ $\xmapsto{\mathcal{Q}}$ $\sin((2s+1)t)\frac{1}{\sqrt{k}}\sum_{f(p_i)\geq \theta}\ket{i}\ket{p_i} +\cos((2s+1)t)\frac{1}{\sqrt{N-k}}\sum_{f(p_i)< \theta}\ket{i}\ket{p_i}$;\\
			$\ket{\psi}$ $\xmapsto{\mathcal{F}}$ $\sin((2s+1)t)\frac{1}{\sqrt{k}}\sum_{f(p_i)\geq \theta}\ket{i}\ket{f(p_i)} +\cos((2s+1)t)\frac{1}{\sqrt{N-k}}\sum_{f(p_i)< \theta}\ket{i}\ket{f(p_i)}$;\\
			$\ket{\psi}$ $\xmapsto{\mathcal{G}_\theta}$ $-\sin((2s+1)t)\frac{1}{\sqrt{k}}\sum_{f(p_i)\geq \theta}\ket{i} +\cos((2s+1)t)\frac{1}{\sqrt{N-k}}\sum_{f(p_i)< \theta}\ket{i}$;\\
			$\ket{\psi}$ $\xmapsto{H\mathcal{R}H}$ $\sin((2s+3)t)\frac{1}{\sqrt{k}}\sum_{f(p_i)\geq \theta}\ket{i} +\cos((2s+3)t)\frac{1}{\sqrt{N-k}}\sum_{f(p_i)< \theta}\ket{i}$;\\
		}
		\tcp*[l]{Post-selection}
		$\ket{\psi}$ $=$ $\sin((2j+1)t)\frac{1}{\sqrt{k}}\sum_{f(p_i)\geq \theta}\ket{i} +\cos((2j+1)t)\frac{1}{\sqrt{N-k}}\sum_{f(p_i)< \theta}\ket{i}$;\\
		$\ket{\psi}$ $\xmapsto{\mathcal{Q}}$ $\sin((2j+1)t)\frac{1}{\sqrt{k}}\sum_{f(p_i)\geq \theta}\ket{i}\ket{p_i} +\cos((2j+1)t)\frac{1}{\sqrt{N-k}}\sum_{f(p_i)< \theta}\ket{i}\ket{p_i}$;\\
		$\ket{\psi}$ $\xmapsto{\mathcal{F}}$ $\sin((2j+1)t)\frac{1}{\sqrt{k}}\sum_{f(p_i)\geq \theta}\ket{i}\ket{f(p_i)} +\cos((2j+1)t)\frac{1}{\sqrt{N-k}}\sum_{f(p_i)< \theta}\ket{i}\ket{f(p_i)}$;\\
		Append an auxiliary qubit $\ket{\phi}=\ket{0}$;\\
		$\ket{\psi}\ket{\phi}$ $=$ $\sin((2j+1)t)\frac{1}{\sqrt{k}}\sum_{f(p_i)\geq \theta}\ket{i}\ket{f(p_i)}\ket{0} +\cos((2j+1)t)\frac{1}{\sqrt{N-k}}\sum_{f(p_i)< \theta}\ket{i}\ket{f(p_i)}\ket{0}$;\\
		$\ket{\psi}\ket{\phi}$ $\xmapsto{\mathcal{O}_\theta}$ $\sin((2j+1)t)\frac{1}{\sqrt{k}}\sum_{f(p_i)\geq \theta}\ket{i}\ket{f(p_i)}\ket{1} +\cos((2j+1)t)\frac{1}{\sqrt{N-k}}\sum_{f(p_i)< \theta}\ket{i}\ket{f(p_i)}\ket{0}$;\\
		Measure $\ket{\phi}$;\\
		\lIf{$\ket{\phi}=1$}{
			\Return{\ket{\psi}}\label{Line:returnpsi}
		}
		$m$ $\leftarrow$ $m\cdot \frac{4}{3}$;\\
	}
	\Return{NULL}\label{Line:returnnull}
\end{algorithm}

\subsubsection{Analysis} \label{sec:ana}
\leavevmode\\

Combining amplitude amplification and post-selection, we obtain the algorithm for QQPQ$_\theta$ problems. Motivated by \cite{boyer1998tight}, we use the same loop method to solve the problem that $k$ is unknown. Algorithm \ref{alg-1} shows the procedure. 

Figure \ref{fig:QQPQ} shows an illustration of the quantum circuit. $q_9$, $q_8$ and $q_7$ store the indices of the tuples. $q_6$ and $q_5$ store the first dimension of the tuple. $q_4$ and $q_3$ store the second dimension of the tuple. $q_2$ and $q_1$ store the utility of the tuple. $q_0$ is the auxiliary qubit in the post-selection. First, we use three Hadamard gates to initialize the indices. After this step, $q_9$, $q_8$ and $q_7$ all the indices from $0$ to $7$. Then, we start to iterate amplitude amplification for $j$ times. We first use $\mathcal{Q}$ to read the attributes of the tuples and store the information in $q_6, q_5, q_4$ and $q_3$. Then, we use $\mathcal{F}$ to calculate the utility and store the result in $q_2$ and $q_1$. After this step, $\mathcal{G}_\theta$ is applied to flip the amplitudes of tuples with utilities larger than $\theta$. Then, we re-apply $\mathcal{Q}$ and $\mathcal{F}$ to turn $q_6, q_5, q_4, q_3, q_2, q_1$ into the initial states so that we can reuse these qubits in the next iterations. The last step in amplitude amplification is to use $H\mathcal{R}H$ to amplify the flipped amplitudes. After $j$ iterations, we start the post-selection. After applying $\mathcal{Q}, \mathcal{F}$ and $\mathcal{O}_\theta$ in order, we measure $q_0$. If $q_0$ is $1$, then we obtain the answer in $q_9, q_8$ and $q_7$.

Note that the answer of a QQPQ$_\theta$ problem can be empty, which corresponds to the case that no tuple has utility higher than $\theta$. If the answer is empty, the algorithm returns $NULL$ as shown in Line \ref{Line:returnnull}. Otherwise, it returns the superposition of the desired tuples as shown in Line \ref{Line:returnpsi}.  Such a loop method guarantees that the false negative rate is at most $\frac{1}{4}$ \cite{boyer1998tight}, so it can be arbitrarily small if we repeat the process for a constant time. To analyze the complexity, we have the following Theorem \ref{theo1}.
\begin{theorem}\label{theo1}
	The QQPQ$_\theta$ algorithm needs $\frac{9}{2}\sqrt{\frac{N}{k}}$ IOs on average to answer a query.
\end{theorem}
\begin{proof}
	Proved by \cite{boyer1998tight}, the expected number of iterations needed is at most $\frac{9}{2}\sqrt{\frac{N}{k}}$. Since each iteration performs $1$ QRAM read, the algorithm needs $\frac{9}{2}\sqrt{\frac{N}{k}}$ IOs in total.
\end{proof}
The $O(\sqrt{\frac{N}{k}})$ complexity shows the advantage of the quantum algorithms, since a larger $k$ leads to a lower computational cost, which can never be achieved by any classical algorithms.

\subsection{Classical-Output Threshold-Based Quantum Preference Query}\label{sec:CQPQT}
Assume there are $k$ desired tuples with utilities higher than $\theta$. In this problem, the algorithm needs to return a list $\{l_0, l_1,\cdots, l_{k-1}\}$ such that $(p_{l_0}, p_{l_1},\cdots, p_{l_{k-1}})$ is the desired tuples. 

We propose the CQPQ$_\theta$ algorithm to solve this problem. Our intuition is that we utilize the superposition returned by QQPQ$_\theta$ algorithm to find the desired tuples in classical bits one by one in a number of iterations. In each $t$-th iteration, the main steps are as follows.
\begin{itemize}
	\item Step 1: We use QQPQ$_\theta$ algorithm to obtain a superposition of the desired tuples.
	\item Step 2: We measure the returned qubits. We will obtain a random desired tuple $l_t$.
	\item Step 3: We assign a special mark $dummy$ to the $l_t$-th tuple, such that it will always be assigned the lowest utility by the utility function $f$.
\end{itemize}
In the $t$-th iteration, there are $k-t$ non-dummy tuples with utilities higher than $\theta$. We use QQPQ$_\theta$ algorithm to obtain a quantum result $\frac{1}{\sqrt{k-t}}\sum_{f(p_i)\geq \theta}\ket{i}\ket{f(p_i)}$. If we measure the result, we randomly obtain one of the $k-s$ desired tuples with probability $\frac{1}{k-t}$. Then, we assign $dummy$ to it so that it will not be marked as desired tuples in the subsequent iterations. Algorithm \ref{alg-2} shows the procedure.

\begin{algorithm}[htbp]
\SetKwFunction{QCTK}{CQPQ$_\theta$}
\SetKwFunction{QTK}{QQPQ$_\theta$}
\SetKwFunction{Break}{break}
\caption{\protect\QCTK{f, $\theta$}} \label{alg-2}
\KwIn{A utility function $f$ and a utility threshold $\theta$.}
\KwOut{A list $\{l_0, l_1,\cdots, l_{k-1}\}$ such that for each $i\in [0, k-1]$, $f(p_{l_i})>\theta$, where $k$ is the number of tuples with utilities higher than $\theta$.}
	$L$ $\leftarrow$ $\emptyset$;\\
	$t$ $\leftarrow$ $0$;
	\Repeat{$\ket{\psi}=NULL$}{
		$\ket{\psi}$ $\leftarrow$ \QTK{f, $\theta$};\\
		\If{$\ket{\psi} \neq NULL$}{
			$l_t$ $\leftarrow$ the first $n$ bits of $\ket{\psi}$ after measurement;\\
			$\mathcal{Q}[l_t]$ $\leftarrow$ $dummy$;\\
			$L$ $\leftarrow$ $L\cup \{l_t\}$;\\
                $t = t + 1$; \\
		}
	}
	\Return{$L$}
\end{algorithm}

To analyze the complexity, we have the following Theorem \ref{theo2}.
\begin{theorem}\label{theo2}
	The CQPQ$_\theta$ algorithm needs $9\sqrt{Nk}$ IOs on average to answer a query.
\end{theorem}
\begin{proof}
	By Theorem \ref{theo1}, in the $t$-th iteration, QQPQ$_\theta$ algorithm needs $\frac{9}{2}\sqrt{\frac{N}{k-t}}$ IOs since there are $k-t$ desired tuples. Therefore, the total number of IOs needed is 
\begin{align*}
&\sum_{t=0}^{k-1}\frac{9}{2}\sqrt{\frac{N}{k-t}}
\leq \frac{9}{2}\int_{0}^{k}\sqrt{\frac{N}{k-x}} dx\\
=&\frac{9}{2}\left.(-2\sqrt{N\cdot(k-x)})\right|^k_0
=9\sqrt{Nk}
\end{align*}
\end{proof}
In most real-world cases, we have $k<<N$. Therefore, the $O(\sqrt{Nk})$ complexity has an asymptotically quadratic improvement than classical algorithms that need at least $\Omega(N)$ time.

\subsection{Classical-Output Top-$k$ Quantum Preference Query}\label{sec:CQPQK}
In this problem, the algorithm need to return $k$ tuples with the highest utilities. The main task is to determine the $k$-th highest utility. We propose the CQPQ$_k$ algorithm to solve this problem.

In the algorithm, we maintain a min-priority queue. The main steps are as follows.
\begin{itemize}
	\item Step 1: We randomly pick $k$ tuples and insert them into the min-priority queue. We also remove these $k$ tuples from the dataset. 
	\item Step 2: We use QQPQ$_\theta$ algorithm to obtain a superposition $\ket{\psi}$ of tuples among the remaining tuples in the dataset with utilities higher than the minimum utility in the min-priority queue. If $NULL$ is returned, then the min-priority queue exactly contains the $k$ desired tuples. Otherwise, we execute Step 3.
	\item Step 3: We measure the returned qubits $\psi$ which collapse to an resulting tuple, says $p_{\psi}$. We then pop the tuple in the min-priority queue with the minimum utility and push $p_{\psi}$ into the queue. We also remove $p_{\psi}$ from the dataset. Then, we go back to Step 2.
\end{itemize}
In each iteration from Step 2 to Step 4, we update the min-priority queue with an tuple with a higher utility. When QQPQ$_\theta$ algorithm returns $NULL$ in Step 2, the min-priority queue cannot be updated any more, so the $k$ tuples in the min-priority queue are the desired $k$ tuples. Algorithm \ref{alg-3} shows the procedure.

\begin{algorithm}[htbp]
\SetKwFunction{QCTK}{CQPQ$_\theta$}
\SetKwFunction{QTK}{QQPQ$_\theta$}
\SetKwFunction{CQK}{CQPQ$_k$}
\SetKwFunction{Break}{break}
\caption{\protect\CQK{$f$, $k$}} \label{alg-3}
\KwIn{A utility function $f$ and a positive integer $k$.}
\KwOut{A list $\{l_0, l_1,\cdots, l_{k-1}\}$ such that $\{p_{l_0}, p_{l_1},p_{l_{k-1}}\}$ are the $k$ tuples with the highest utilities.}
	$Q$ $\leftarrow$ $\emptyset$;\tcp*[r]{Initialize the min-priority queue}
        $D'$ $\leftarrow$ $\{0, 1,\cdots, N-1\}$; \\
	\For{$i$ $\leftarrow$ $0$ \KwTo $k-1$}{
		$item$ $\leftarrow$ a random number in $D'$;\\
            Remove $item$ from $D'$; \\
		Push $item$ into $Q$;\\
	}
	\Repeat{$\ket{\psi}=NULL$}{
		$item$ $\leftarrow$ The tuple in $Q$ with the minimum utility;\\
		$\theta$ $\leftarrow$ $f(p_{item})$;\\
		$\ket{\psi}$ $\leftarrow$ \QTK{f, $\theta$} with dataset $D'$;\\
		\If{$\ket{\psi} \neq NULL$}{
			$nextitem$ $\leftarrow$ the first $n$ bits of $\ket{\psi}$ after measurement;\\
			Pop the tuple with the minimum utility from $Q$;\\
			Push $nextitem$ into $Q$;\\
		}
	}
	\Return{$Q$}
\end{algorithm}
To analyze the complexity, we need to first discuss the probability that each tuple will be inserted into the priority queue. We use $P(i)$ to denote the probability that the tuple with the $i$-th highest utility will be inserted into the min-priority queue. Then, the following Lemma \ref{lemma1} can be obtained.
\begin{lemma}\label{lemma1}
	The probability $P(i)$ depends on $i$ such that 
	$$P(i)=
\left\{
	\begin{array}{ll}
		1 &i\leq k ;\\
		k/i &\mbox{otherwise} .\\
	\end{array}
\right.$$
\end{lemma}
\begin{proof}
If $i\leq k$, then the tuple with the $i$-th highest utility must be inserted into $Q$, since it is one of the desired tuples. Therefore, $P(i)=1$.

Then, we discuss the case $i>k$. If $i=N$, we have $P(i)=k/N$, since we can only insert the tuple in Step 1. Otherwise, the probability contains two parts. The first part is $\frac{k}{N}$, since each tuple will be inserted into $Q$ with probability $\frac{k}{N}$ in Step 1. The second part is the probability of insertion in Step 3. We can observe that from the perspective of the tuple with the $j$-th highest utility, each tuple with a higher utility is identical in QQPQ$_\theta$ algorithm, so they will be inserted with the same probability $\frac{1}{j-1}$. If $i<N$, assume it holds for $i+1$ to $N$, then 
\begin{align*}
P(i)&=\frac{k}{N}+\sum_{j=i+1}^N \frac{1}{j-1}P(j)
=\frac{k}{N}+\sum_{j=i+1}^N \frac{1}{j-1}\cdot \frac{k}{j}\\
&=\frac{k}{N}+ \frac{1}{N-1}\cdot \frac{k}{N}+\cdots+\frac{1}{i}\cdot \frac{k}{i+1}\\
&=\frac{k}{N}+ \frac{k}{N-1}- \frac{k}{N}+\cdots+\frac{k}{i}- \frac{k}{i+1}\\
&=\frac{k}{i}
\end{align*}
\end{proof}
By Lemma \ref{lemma1}, we know the probability that each tuple appears in the min-priority queue, which leads to a QQPQ$_\theta$ query and a min-priority queue insertion. To analyze the complexity, we have the following Theorem \ref{theo3}.
\begin{theorem}\label{theo3}
	The CQPQ$_k$ algorithm needs $\frac{9\pi}{2}\sqrt{Nk}+ k\log_2 k \ln N$ IOs on average to answer a query.
\end{theorem}
\begin{proof}
Since the desired $k$ tuples will not trigger another QQPQ$_\theta$ query and a min-priority queue insertion, we only need to calculate the IO cost needed by the other $N-k$ tuples, so we obtain
\begin{align*}
&\sum_{i=k+1}^N P(i)\cdot \frac{9}{2}\sqrt{\frac{N}{i-k}}+\sum_{i=k+1}^N P(i)\cdot \log_2 k\\
=&\frac{9k\sqrt{N}}{2}\sum_{i=k+1}^N \frac{1}{i}\cdot \sqrt{\frac{1}{i-k}}+k\log_2 k\cdot\sum_{i=k+1}^N \frac{1}{i}\\
\leq & \frac{9k\sqrt{N}}{2}\int_k^N \frac{1}{x}\cdot \sqrt{\frac{1}{x-k}}dx + k\log_2 k \ln N\\
=&\frac{9k\sqrt{N}}{2}\left. (\frac{2}{\sqrt{k}}\arctan (\sqrt{\frac{x-k}{k}})) \right|_k^N + k\log_2 k \ln N\\
\leq & \frac{9\pi}{2}\sqrt{Nk}+ k\log_2 k \ln N.\\
\end{align*}
\end{proof}
By Theorem \ref{theo3}, the complexity of our proposed CQPQ$_k$ algorithm is $O(\sqrt{Nk}+ k\log_2 k \ln N)$. 

\subsection{Quantum-Output Top-$k$ Quantum Preference Query}\label{sec:QQPQK}
In this problem, the algorithm needs to return a superposition of the desired $k$ tuples. However, compared to CQPQ$_k$ problem, there is little hope to obtain further speedup. The reason is that we can only use a classical method as shown in Section \ref{sec:CQPQK} to determine the $k$-th highest utility, due to the impossibility of comparing quantum states \cite{arul2001impossibility}. Therefore, to solve QQPQ$_k$ problem, we propose Algorithm \ref{alg-4} based on QQPQ$_\theta$ algorithm and CQPQ$_k$ algorithm. Obviously, the complexity of QQPQ$_k$ algorithm is $O(\sqrt{Nk})$.
\begin{algorithm}[htbp]
\SetKwFunction{QCTK}{CQPQ$_\theta$}
\SetKwFunction{QTK}{QQPQ$_\theta$}
\SetKwFunction{CQK}{CQPQ$_k$}
\SetKwFunction{QK}{QQPQ$_k$}
\SetKwFunction{Break}{break}
\caption{\protect\QK{$f$, $k$}} \label{alg-4}
\KwIn{A utility function $f$ and a positive integer $k$.}
\KwOut{$\frac{1}{\sqrt{k}}\sum_{f(p_i)\geq \theta}\ket{i}\ket{f(p_i)}$.}
	$L$ $\leftarrow$ \CQK{$f$, $k$};\\
	$\theta$ $\leftarrow$ $\max_{l \in L}f(p_l)$;\\
	\Return{\QTK{$f$, $\theta$}}
\end{algorithm}

\section{Experiment}
\label{sec:experiments}
In this section, we show our experimental results on our quantum preference queries. The study of the real-world quantum supremacy \cite{arute2019quantum, boixo2018characterizing, terhal2018quantum}, which is to confirm that a quantum computer can do tasks faster than classical computers, is still a developing topic in the quantum area. We do not aim at verifying quantum supremacy in this paper, but we believe it will be verified in a future quantum computer.

To implement the quantum algorithms and conduct scalability tests, we mainly use C++ to perform the quantum simulations, since existing quantum simulators (e.g., Qiskit \cite{wille2019ibm} and Cirq \cite{cirq}) cannot simulate QRAM effectively. Nevertheless, for a comprehensive comparison, we also implement our algorithms by applying an existing widely-used library \cite{qramLibrary} for QRAM simulation (which is implemented with the Q\# language) and make comparison with baselines implemented with Q\# as well. All algorithms are implemented in a classical machine with 3.6GHz CPU and 32GB memory.

\smallskip\noindent\textbf{Datasets.} We use synthetic and real datasets that were commonly used in existing studies \cite{borzsony2001skyline, tang2019creating, tang2021m}. The synthetic datasets include \emph{anti-correlated} (ANTI), \emph{correlated} (CORR), \emph{independent} (INDE) which have different data distributions. The real datasets include \emph{HOTEL} \cite{hotelDataset}, \emph{HOUSE} \cite{houseDataset}, \emph{NBA} \cite{NBADataset}. HOTEL includes 419k tuples with 4 attributes (i.e., price and numbers of stars, rooms and facilities). HOUSE includes 315k tuples with 6 attributions (i.e., gas, electricity, water, heating, insurance and property tax). NBA includes 21.9k tuples with 8 attributes (i.e., games, rebounds, assists, steals, blocks, turnovers, personal fouls and points).

\smallskip\noindent\textbf{Measurement.} In this paper, we mainly evaluate the \emph{number of memory accesses}, which corresponds to the number of IOs in traditional searches. In the quantum algorithms, a QRAM read operation is counted as 1 IO. In the classical algorithms, a page access is counted as 1 IO. IOs cannot be regarded as the real execution time, since we have no any information about the real implementation of a physical QRAM, but they can reflect the potential of the algorithms. Furthermore, we also compare the execution time of algorithms implemented with Q\# for additional verification. Note that even with existing QRAM simulation in Q\#, the execution time may still be different from that on real quantum computers.

\smallskip\noindent\textbf{Algorithms.} 
Since there is no classical quantum algorithm returning quantum output, we focused on the classical-output algorithms. We evaluate our algorithms for the top-$k$ and threshold-based queries (denoted by QPQ$_k$ and QPQ$_\theta$, respectively) with existing classical algorithms \emph{Quick Selection} \cite{cormen2022introduction}, \emph{$\tau$-Level Index} \cite{zhang2022t} and \emph{Linear Scan}. Specifically, we compare QPQ$_k$ with Quick Selection and $\tau$-Level Index, and compare QPQ$_\theta$ with Linear Scan (because Quick Selection and $\tau$-Level Index can only handle the top-$k$ queries). Note that $\tau$-Level Index needs preprocessing which costs up to billions of IOs. However, we do not count the IOs in its preprocessing in comparison and only focused on the IO costs for queries.

\smallskip\noindent\textbf{Parameter Setting.} We evaluated the performance of algorithms by varying several parameters: (1) parameter $k$, used in QPQ$_k$; (2) parameter $\theta$, used in QPQ$_\theta$; (3) the number of dimensions $d$; (4) the number of tuples in the dataset $N$; (5) the category (ANTI, CORR and INDE), which reflects the data distribution in the synthetic datasets. Following \cite{zhang2022t}, $n=500$K, $d=4$, $k=10$ by default. For each experimental setting, we randomly generate 100 queries with different utility functions and report the average measurement.

In the following, we show our experimental results.

\begin{figure*}[tbp]
\vspace{-2mm}
  \centering
  \subfigure{
  \includegraphics[width=0.7\textwidth]{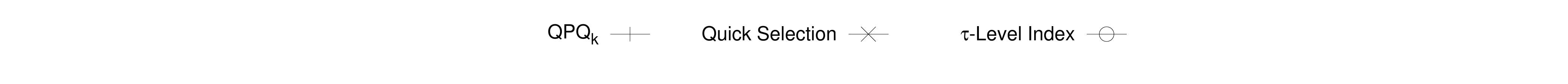}}
  \\
  \vspace{-5mm}
  \setcounter{subfigure}{0}
  \subfigure[ANTI]{
    \label{fig:1-1}
    \includegraphics[width=0.23\textwidth]{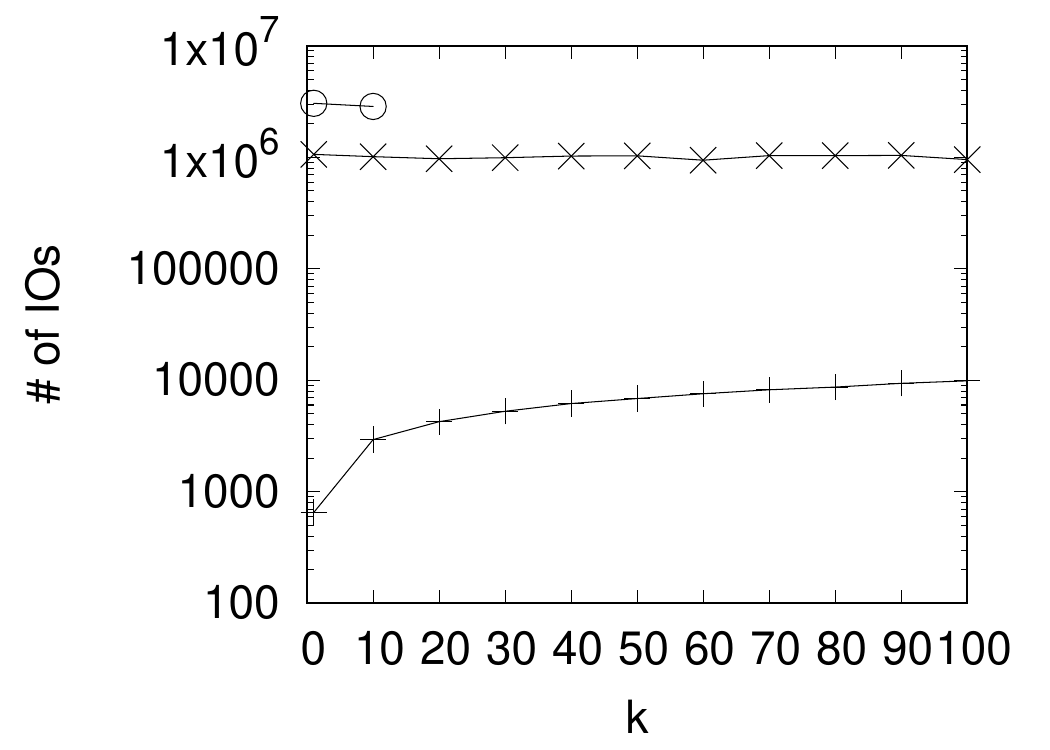}}
    \subfigure[HOTEL]{
    \label{fig:1-2}
    \includegraphics[width=0.23\textwidth]{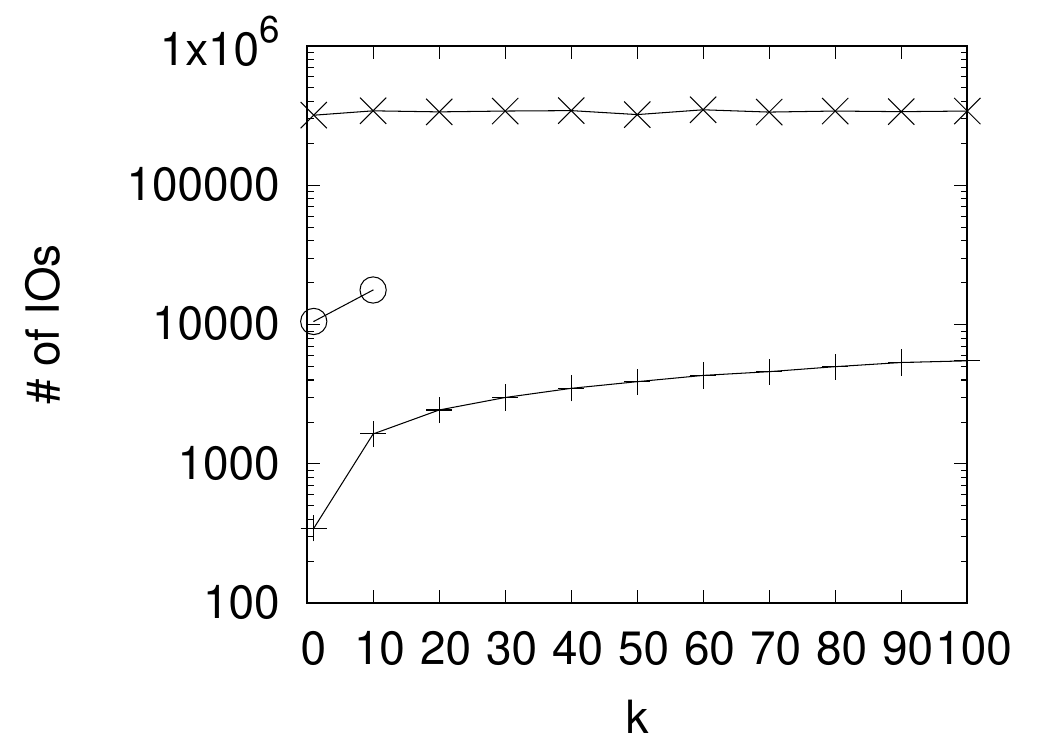}}
  \subfigure[HOUSE]{
    \label{fig:1-3}
    \includegraphics[width=0.23\textwidth]{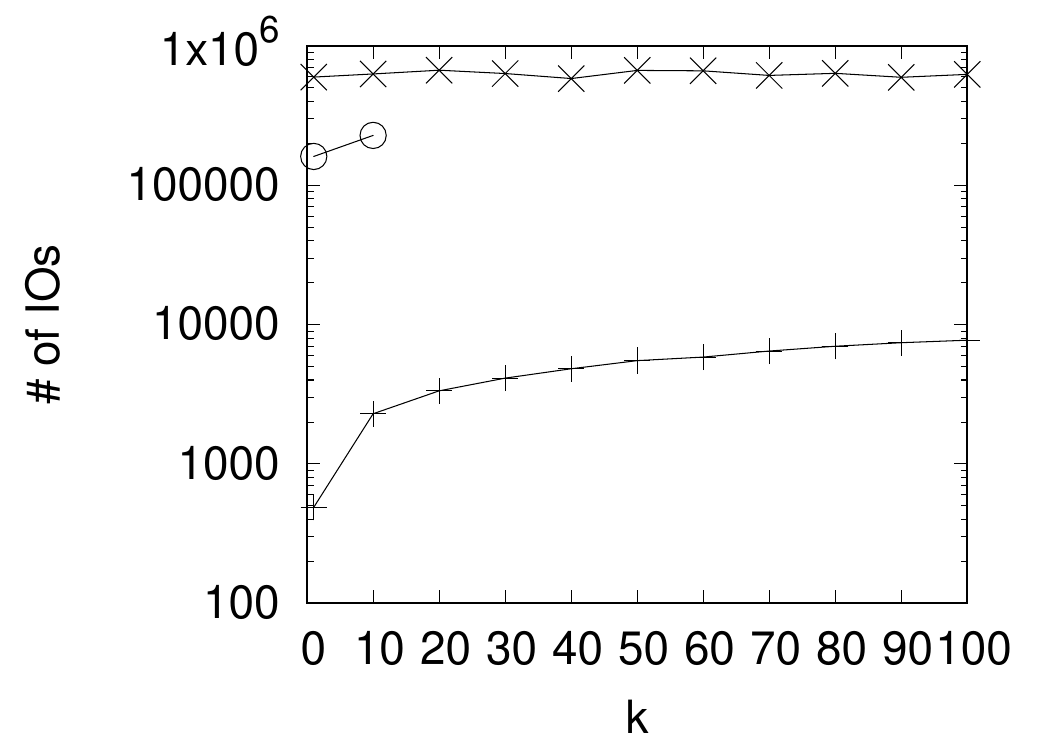}}
    \subfigure[NBA]{
    \label{fig:1-4}
    \includegraphics[width=0.23\textwidth]{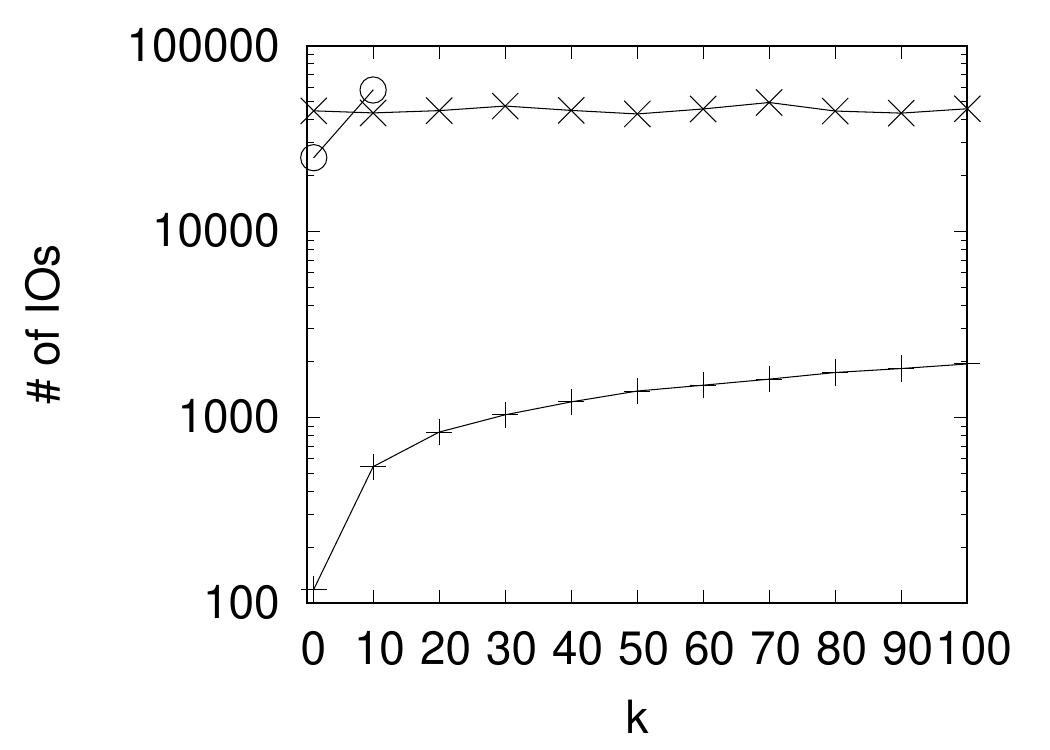}}
  \caption{The Effect of $k$ for the CQPQ$_k$ Queries}
  \label{fig:1}
\end{figure*}

\begin{figure*}[tbp]
\vspace{-4mm}
  \centering
  \subfigure{
  \includegraphics[width=0.8\textwidth]{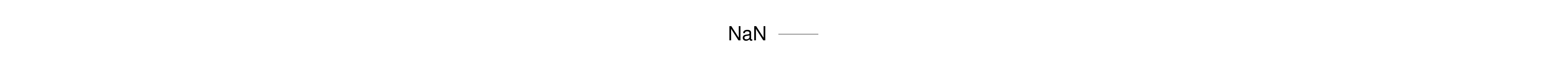}}
  \\
  \vspace{-5mm}
  \setcounter{subfigure}{0}
  \subfigure[Vary $d$]{
    \label{fig:2-1}
    \includegraphics[width=0.25\textwidth]{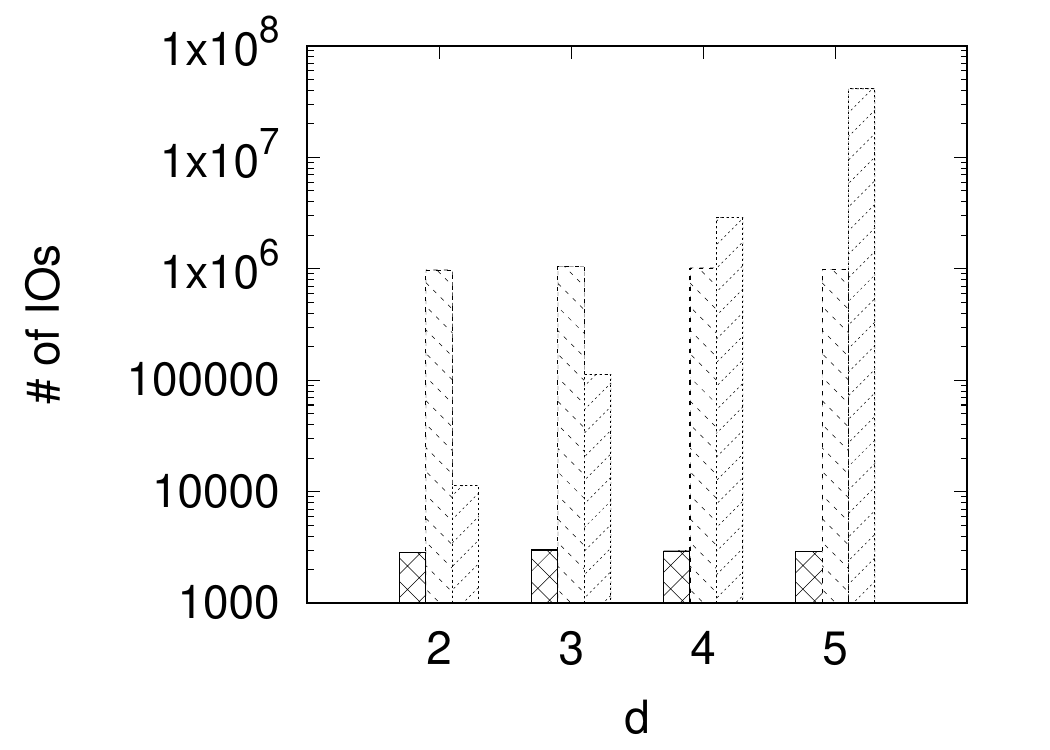}}
    \subfigure[Vary $N$]{
    \label{fig:2-2}
    \includegraphics[width=0.25\textwidth]{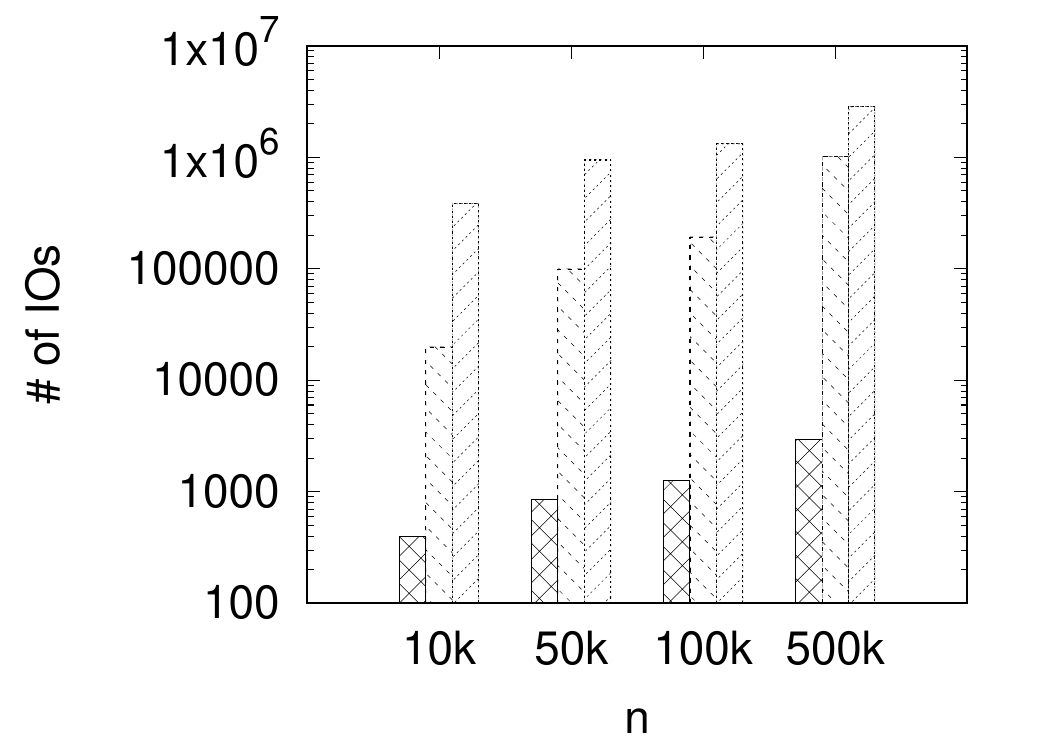}}
  \subfigure[Vary dataset]{
    \label{fig:2-3}
    \includegraphics[width=0.25\textwidth]{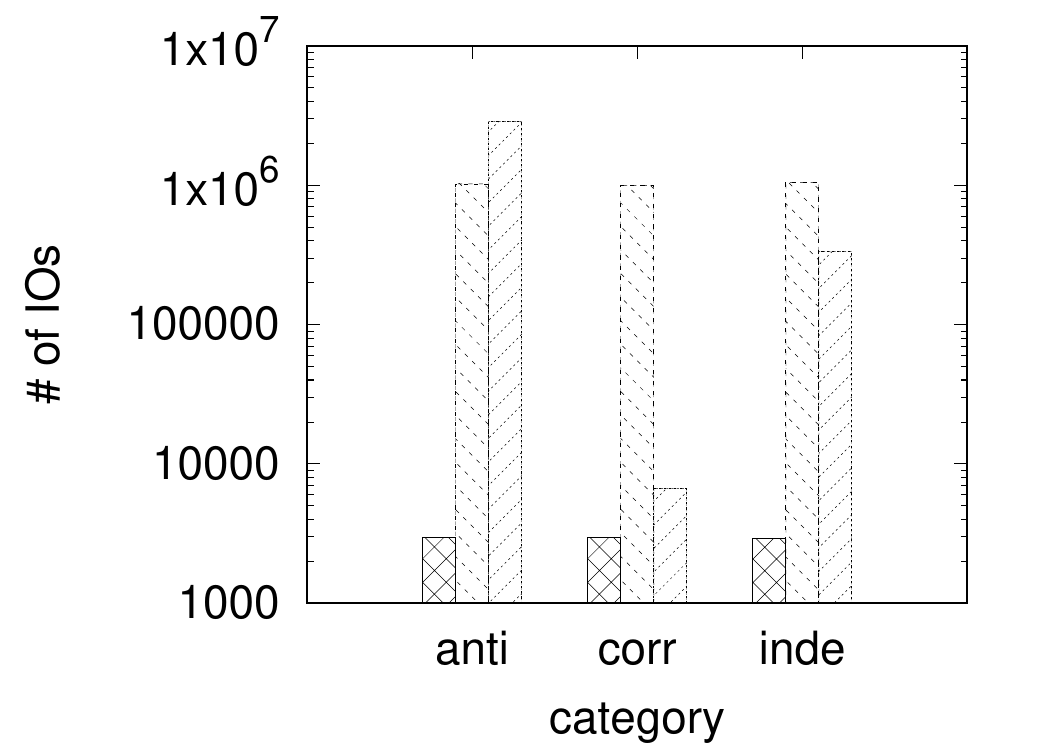}}
  \caption{The Effect of $d$, $N$ and dataset category for the CQPQ$_k$ Queries}
  \label{fig:2}
\end{figure*}

\begin{figure*}[tbp]
\vspace{-4mm}
  \centering
  \subfigure{
  \includegraphics[width=0.7\textwidth]{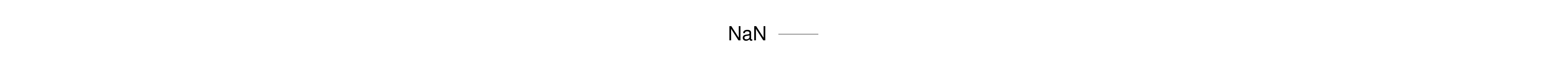}}
  \\
  \vspace{-5mm}
  \setcounter{subfigure}{0}
  \subfigure[ANTI]{
    \label{fig:3-1}
    \includegraphics[width=0.23\textwidth]{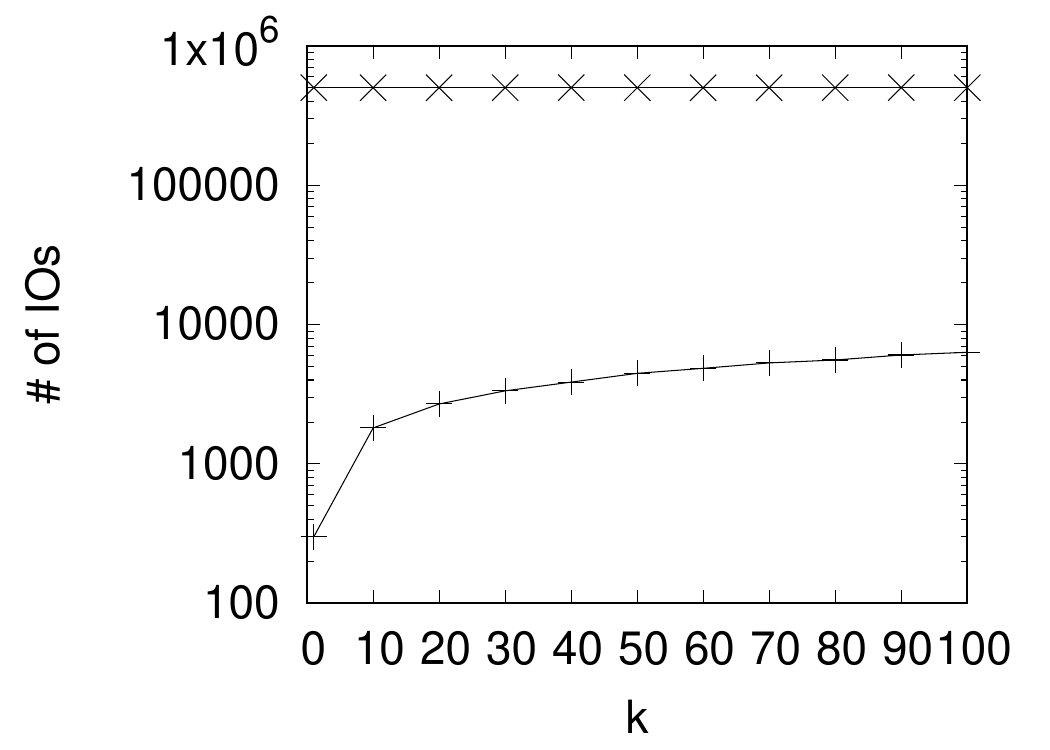}}
    \subfigure[HOTEL]{
    \label{fig:3-2}
    \includegraphics[width=0.23\textwidth]{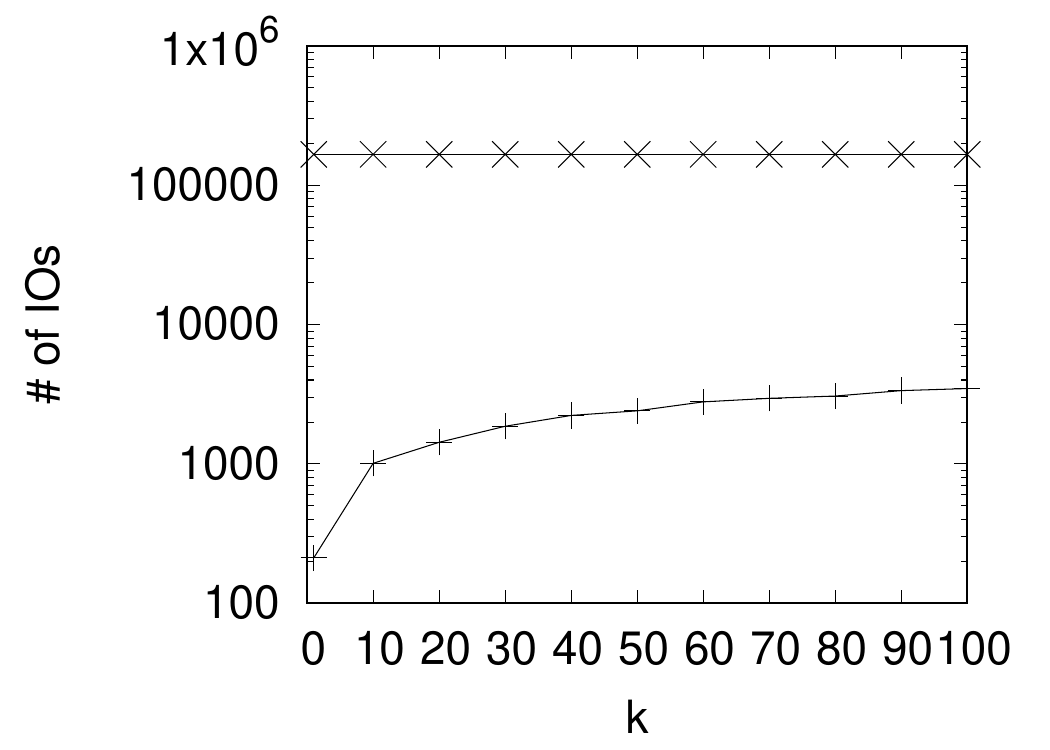}}
  \subfigure[HOUSE]{
    \label{fig:3-3}
    \includegraphics[width=0.23\textwidth]{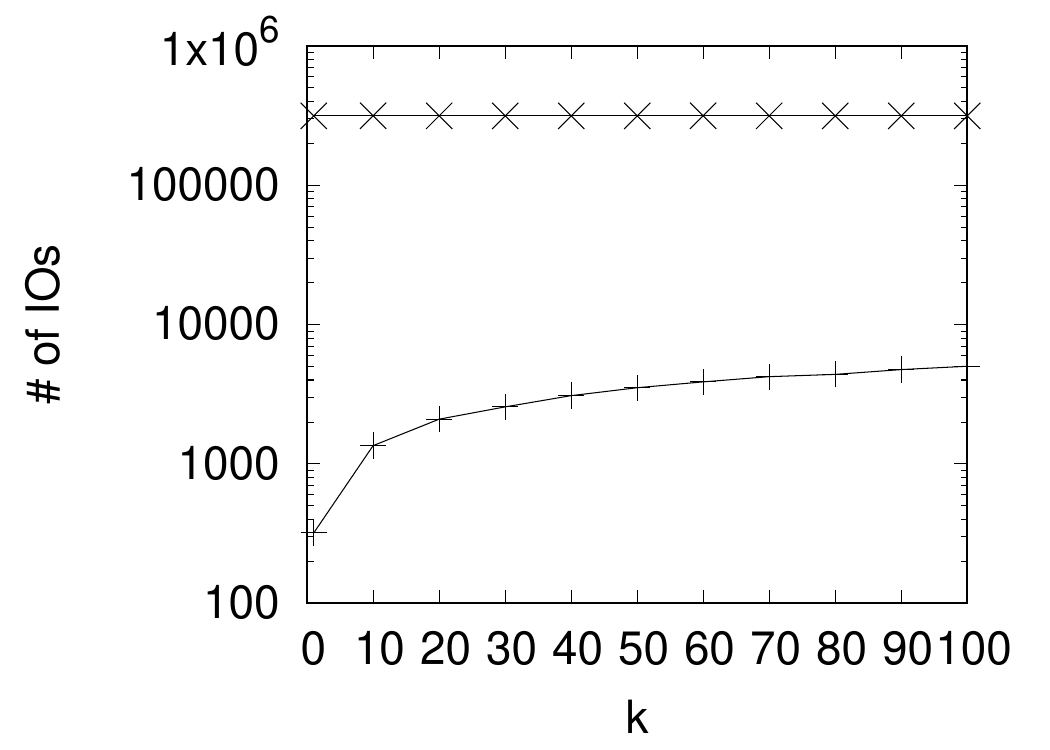}}
    \subfigure[NBA]{
    \label{fig:3-4}
    \includegraphics[width=0.23\textwidth]{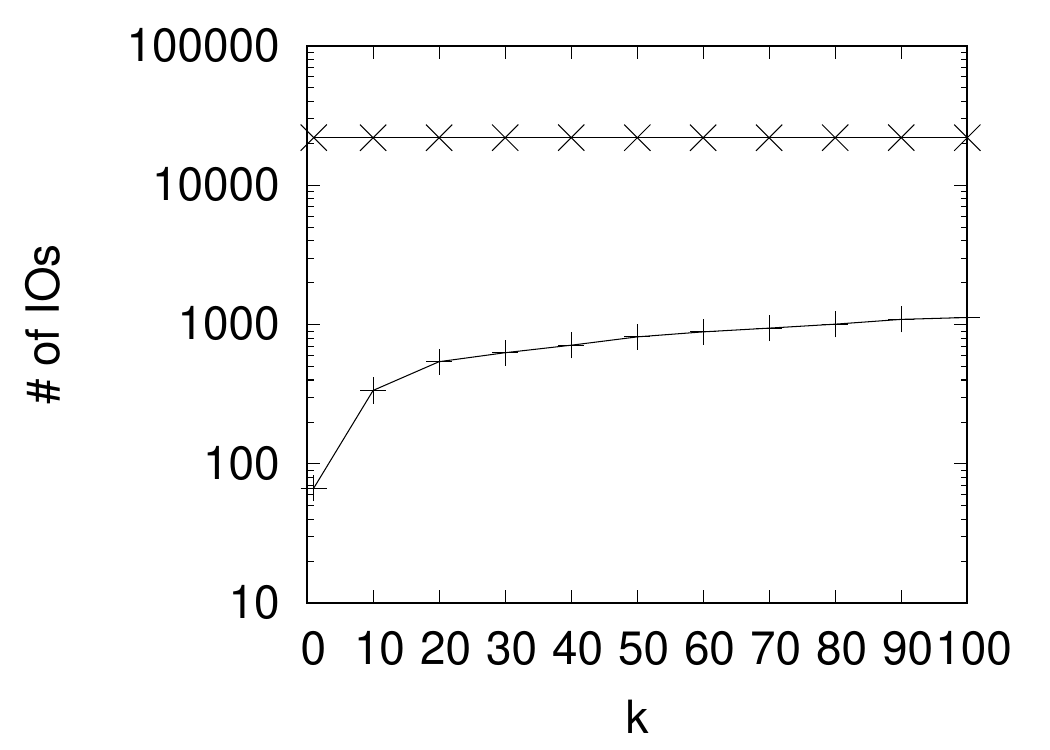}}
  \caption{The Effect of $k$ for the CQPQ$_{\theta}$ Queries}
  \label{fig:3}
\end{figure*}

\begin{figure*}[tbp]
\vspace{-4mm}
  \centering
  \subfigure{
  \includegraphics[width=0.8\textwidth]{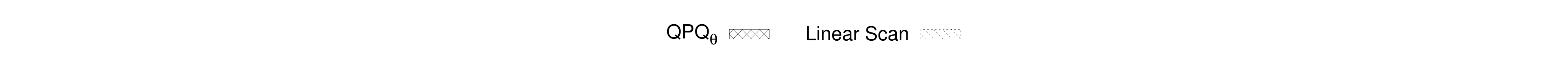}}
  \\
  \vspace{-5mm}
  \setcounter{subfigure}{0}
  \subfigure[Vary $d$]{
    \label{fig:4-1}
    \includegraphics[width=0.25\textwidth]{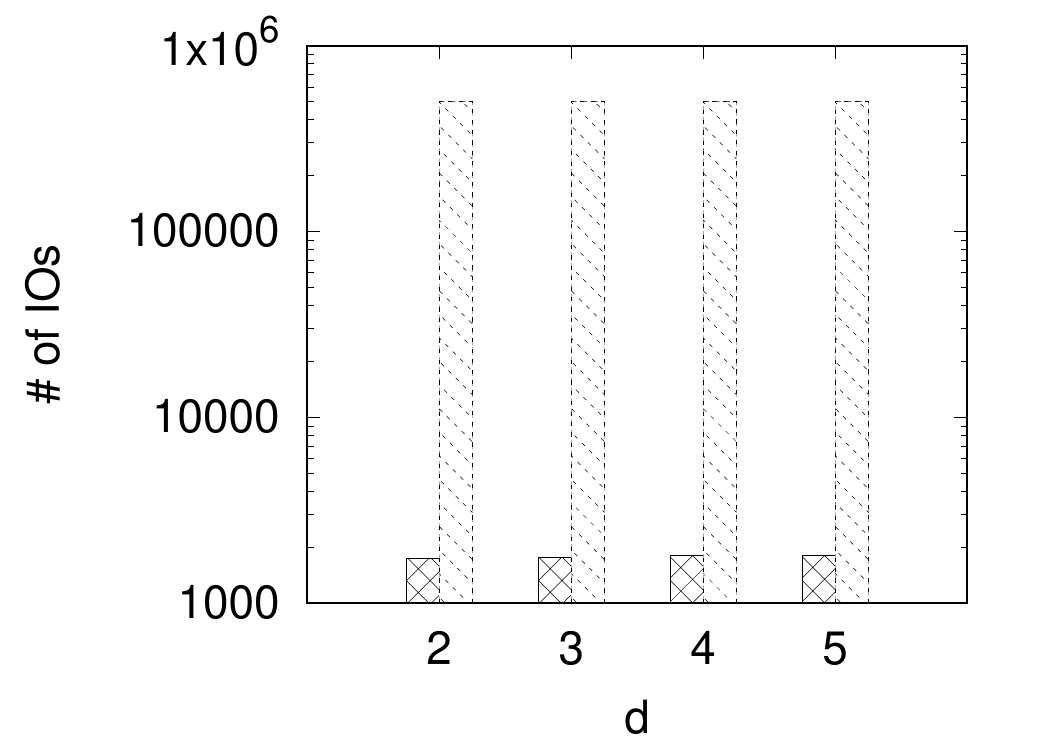}}
    \subfigure[Vary $N$]{
    \label{fig:4-2}
    \includegraphics[width=0.25\textwidth]{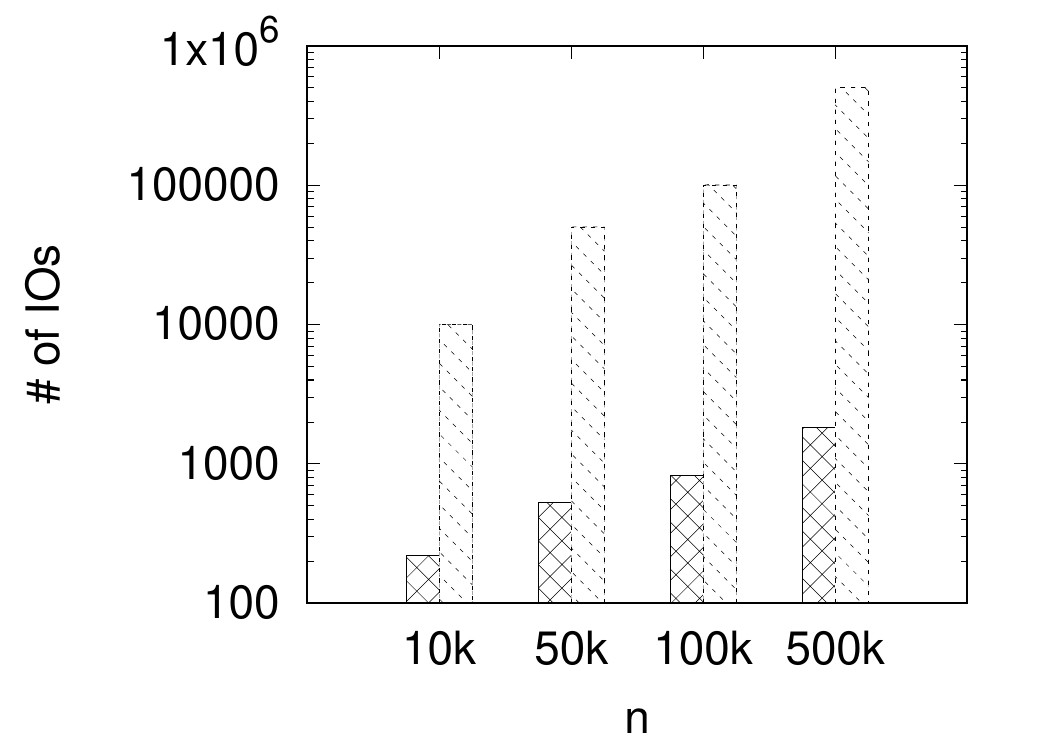}}
  \subfigure[Vary dataset]{
    \label{fig:4-3}
    \includegraphics[width=0.25\textwidth]{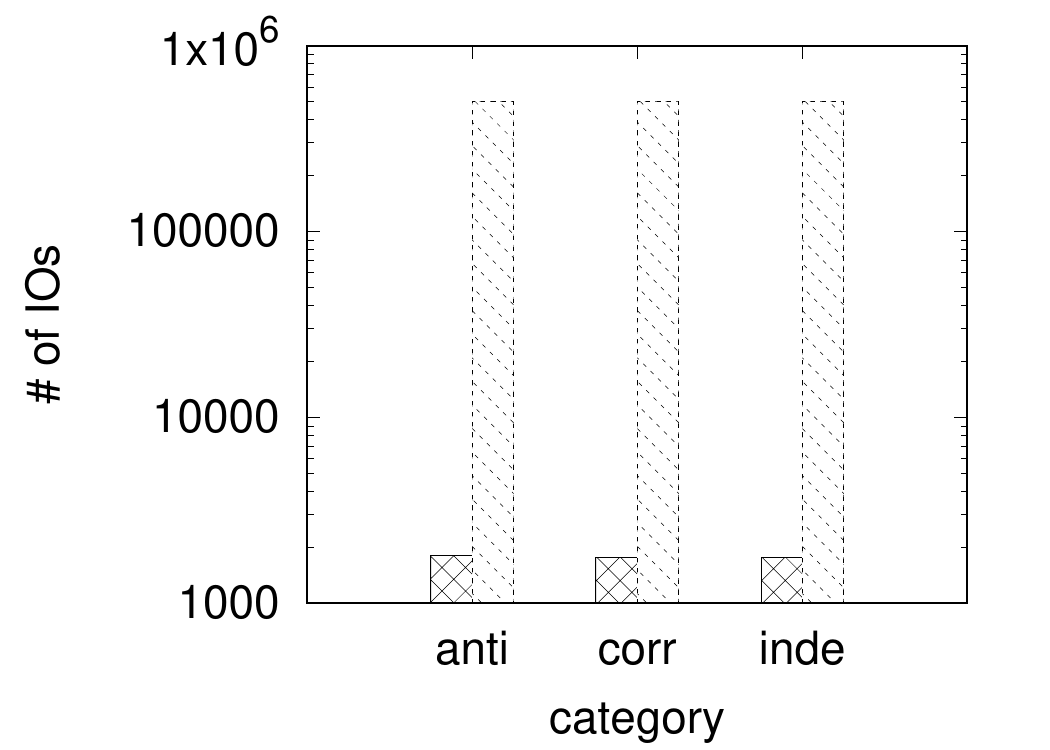}}
  \caption{The Effect of $d$, $N$ and dataset category for the CQPQ$_{\theta}$ Queries}
  \label{fig:4}
\end{figure*}

\begin{figure}[tbp]
\vspace{-4mm}
  \setcounter{subfigure}{0}
  \subfigure[CQPQ$_k$ Queries]{
    \label{fig:4-1}
    \includegraphics[width=0.48\linewidth]{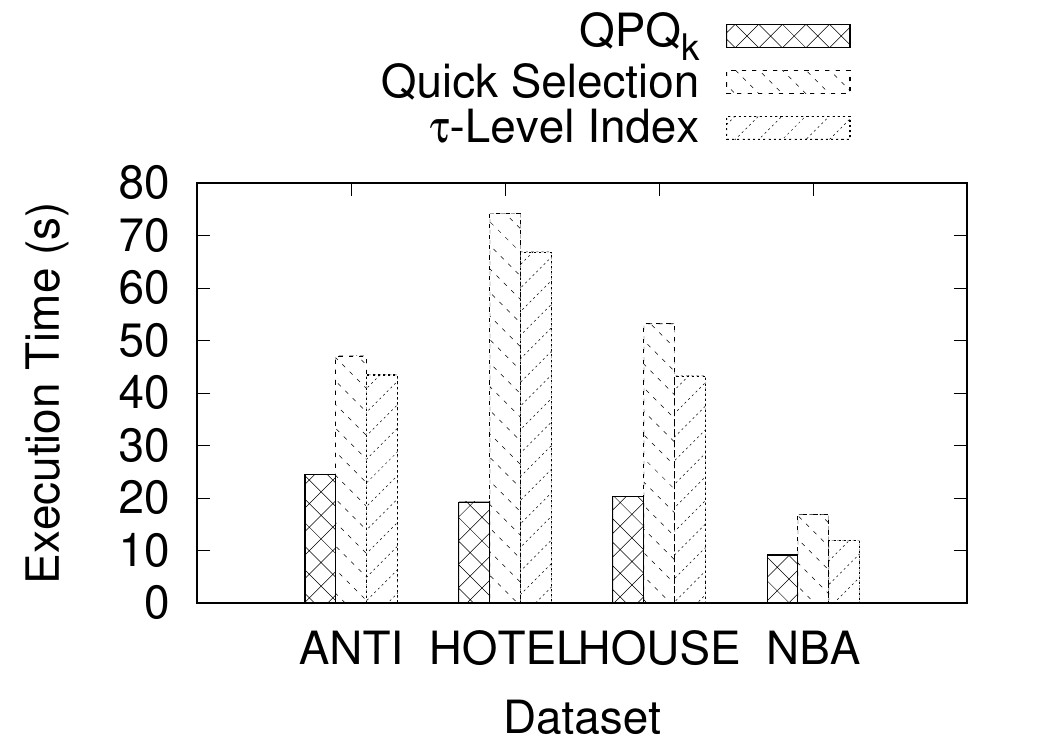}}
    \subfigure[CQPQ$_{\theta}$ Queries]{
    \label{fig:4-2}
    \includegraphics[width=0.48\linewidth]{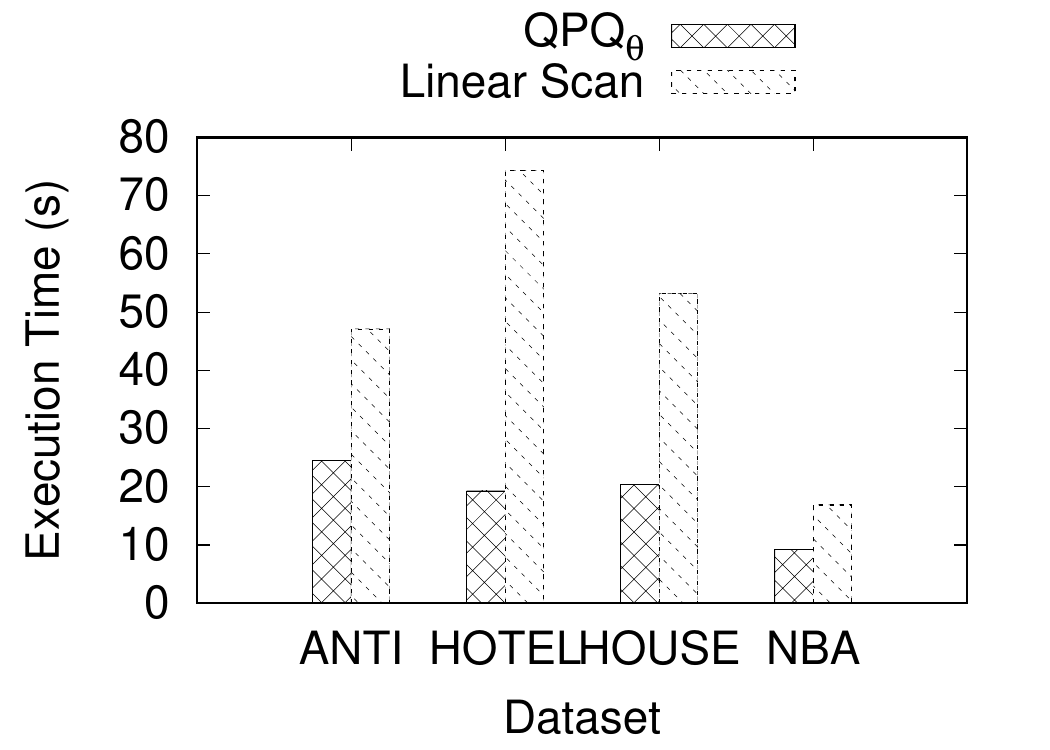}}
  \caption{Execution Time Comparison}
  \label{fig:5}
\end{figure}

\smallskip\noindent\textbf{\underline{Results on Top-$k$ Preference Queries.}}
We first show the comparison results of top-$k$ preference queries (i.e., the CQPQ$_k$ queries).

We first study the effect of $k$. In Figure \ref{fig:1}, we show the results of dataset \emph{ANTI} (the representative synthetic dataset) and the three real datasets. As shown in Figure \ref{fig:1}(a), for dataset \emph{ANTI}, our proposed quantum algorithm QPQ$_k$ has significantly smaller number of IOs than the baselines, which is consistent with the theoretical improvement of our QPQ$_k$ algorithm. As $k$ increases, the IO cost of our QPQ$_k$ grows more obviously first but very slowly for larger $k$, which verifies the theoretical sub-linear growth of our QPQ$_k$ to $k$. In comparison, although baseline Quick Selection has similar number of IOs for varied $k$, the number of IOs is very large (i.e., around $10^6$) where our QPQ$_k$ need IO cost of only around $10^3$ for $k = 1$ and less than $10^4$ for the largest $k$ (i.e., $k = 100$). Baseline $\tau$-Level Index also has very large IO cost. Moreover, $\tau$-Level Index cannot be executed for $k > 10$ due to too long preprocessing time, since the preprocessing time increases exponentially with $k$. We obtain similar results on all the three real datasets, as shown in Figure \ref{fig:1}(a), (b) and (c).

Then, we also study the effect of the number of dimensions $d$, the dataset size $N$ and the different categories of synthetic datasets (where the default synthetic dataset if \emph{ANTI}). As shown in Figure \ref{fig:2}(a), $d$ has little impact on the IO costs of our QPQ$_k$ algorithm and Quick Selection, but the IO cost of $\tau$-Level Index grows significantly with $d$ because the geometric processing of hyperplanes depends heavily on $d$. Still, our QPQ$_k$ has much smaller number of IOs than both baselines. When the dataset size $N$ increases (as illustrated in Figure \ref{fig:2}(b)), all the algorithms have increased IO cost due to more tuples to be processed. Clearly, our QPQ$_k$ algorithm has significantly smaller IO cost (due to utilizing the quantum parallelism) compared with both baselines. As shown in Figure \ref{fig:2}(c), the IO costs of our QPQ$_k$ and Quick Selection also do not depend on the dataset category, while $\tau$-Level Index is more sensitive to dataset category. It indicates that $\tau$-Level Index could have improvement on certain types of data distribution only, but our QPQ$_k$ achieves stably superior performance on IO cost.

\smallskip\noindent\textbf{\underline{Results on Threshold-based Preference Queries.}}
Next, we also show the comparison results of threshold-based preference queries (i.e., the CQPQ$_{\theta}$ queries).

Since the performance is largely affected by the size of the result set (i.e., $k$), we still vary $k$ instead of vary $\theta$ by pick out the $k$-th highest utility from the dataset as the input $\theta$. As shown in Figure \ref{fig:3}, when $k$ is varied, the results are very similar as those the CQPQ$_k$ queries for all the datasets. In particular, the number of IOs of our QPQ$_{\theta}$ algorithm is no more than 500, or even less than 100 for a relatively smaller dataset (e.g., \emph{NBA}) which are close to $\sqrt{N}$, while the number of IOs of baseline Linear Scan is approximately $N$.

Moreover, when varying $d$, $N$ and the dataset category, our QPQ$_{\theta}$ algorithm also obtains superior results, which are similar to the CQPQ$_k$ queries. Baseline Linear Scan has much higher IO cost since its IO cost is linear to $N$ while our QPQ$_{\theta}$ achieves quadratic improvement.

\smallskip\noindent\textbf{\underline{Execution Time Comparison.}}
We also compare the execution time between our algorithms implemented with existing QRAM simulation in Q\# and baselines implemented also with Q\#.

It can be seen from Figure \ref{fig:5}(a) and (b) that, compared with baselines, our proposed quantum algorithms can still achieve superior speed for different datasets, especially those datasets with larger size (e.g., \emph{ANTI}, \emph{HOTEL} and \emph{HOUSE}). However, the execution time improvement (e.g., from 1.3x to 6x) is much smaller than the IO cost improvement (e.g., 1000x). This is because we count each store and load operation of QRAM as 1 IO, but in the existing implementation of QRAM simulators, the time cost cannot be neglected and could be large. Nevertheless, the time improvement of our quantum algorithms could still benefit the real-world applications with the quantum preference queries.

\smallskip\noindent\textbf{\underline{Summary.}} In conclusion, the quantum algorithm performs far better than the classical algorithms from the perspective of the number of memory accesses. With either input $k$ or $\theta$, our QPQ algorithms are 1000$\times$ faster than its classical competitors in typical settings. With the existing QRAM simulators, our proposed quantum algorithms also achieve less execution time than classical baselines. Therefore, we conclude that the quantum algorithms have the potential to outperform classical algorithms.

\section{Conclusion}
\label{sec:conclusion}
In this paper, we discuss four kinds of QPQ problems: QQPQ$_\theta$, CQOQ$_\theta$, CQPQ$_k$ and QQPQ$_k$. We proposed four quantum algorithms to solve these four problems, respectively. For each quantum algorithm, we give an accuracy analysis of the number of memory accesses needed, which shows that the proposed quantum algorithms are at least quadratically faster than their classical competitors. In our experiments, we did simulations to show that to answer a QPQ problem, the quantum algorithms are up to 1000$\times$ faster than their classical competitors, which proved that QPQ problem could be a future direction of the study of preference query problems. The future direction of this work is to consider the quantum algorithms for other types of database queries (e.g., the skyline queries).

\bibliographystyle{ACM-Reference-Format}
\balance
\bibliography{9_reference}


\begin{appendices}
\end{appendices}

\end{document}